\documentclass[10pt]{article}
 \usepackage{epsfig}
 \usepackage{amsmath}
 \usepackage{graphics}
 \usepackage{latexsym}
 \usepackage{amssymb}

 \def\p{\partial}
 \def\ubar{\bar{u}}
 \def\abar{\bar{a}}

 \def\xb{\mathbf{x}}

 \newtheorem{theorem}{Theorem}[section]

 \newenvironment{proof}[1][Proof]{\begin{trivlist}\item[\hskip \labelsep
{\bfseries #1}]}{\end{trivlist}}
\newcommand{\eref}[1]{(\ref{#1})}

\author{
 Greg Norgard %
    \thanks{Graduate Student,Department of Applied Mathematics.}
 and Kamran Mohseni%
   \thanks{Associate Professor of Aerospace Engineering Sciences; Affiliated faculty in the Applied Mathematics Department.}\\
  {\normalsize\itshape
  University of Colorado, Boulder, Colorado, 80309, US}
 }

\title{An Examination of the Homentropic Euler Equations with Averaged Characteristics}

\begin{document}
\maketitle

\begin{abstract}
This paper examines the properties of the homentropic Euler
equations when the characteristics of the equations have been
spatially averaged. The new equations are referred to as the
characteristically averaged homentropic Euler (CAHE) equations. An
existence and uniqueness proof for the modified equations is given.
The speed of shocks for the CAHE equations are determined.  The
Riemann problem is examined and a general form of the solutions is
presented.  Finally, numerically simulations on the homentropic
Euler and CAHE equations are conducted and the behaviors of the two
sets of equations are compared.
\end{abstract}

\section{Introduction}

This manuscript examines a modification of the homentropic Euler
equations where the characteristics of the equations have been
spatially averaged. The equations examined are
\begin{subequations}
\label{CAHE}
\begin{eqnarray}
\label{CAHEa}
\rho_t+\ubar \rho_x +\rho \frac{\abar}{a} u_x=0\\
\label{CAHEb}
u_t+ \ubar u_x+\frac{a\abar}{\rho}\rho_x=0\\
\label{CAHEc}
\ubar=g \ast u\\
\label{CAHEd} \abar=g \ast a
\end{eqnarray}
\end{subequations}
with $a^2=\gamma \rho^{\gamma-1}.$  These are derived in section
\ref{sectionderivation}.  The inspiration for these equations came
from previous work done on a similarly modified Burgers equation.

The Burgers equation, $u_t+uu_x=0$, is considered a simplistic model
of compressible flow. Multiple independent investigations have been
made into a modified Burgers equation with an averaged convective
velocity
\cite{Mohseni:06l,BhatHS:06a,BhatHS:08a,Norgard:08b,NorgardG:08a,Holm:03a},
\begin{subequations}\label{CFB}
\begin{eqnarray}
u_t +\ubar u_x=0\\
\ubar=g^\alpha \ast u\\
g^\alpha=\frac{1}{\alpha} g(\frac{x}{\alpha}),
\end{eqnarray}
\end{subequations}
where $g$ is an averaging kernel. A brief summary of the
primary results are as follows.

\begin{itemize}
\item[1.] Solutions to Equations (\ref{CFB}) exist and are unique
\cite{Norgard:08b,BhatHS:06a}.

\item[2.] When the initial conditions are $C^1$, the solution  remains $C^1$ for all time
\cite{Norgard:08b}.

\item[3.] For bell shaped initial conditions it was proven that as the averaging approaches
 zero ($\alpha \to 0$), the solutions to Equation (\ref{CFB}) converge
  to the entropy solution of inviscid Burgers Equation \cite{NorgardG:08a}.

\item[4.]  There are a set of discontinuous initial conditions where
solutions to Equation (\ref{CFB}) will converge to a non-entropic
solution, but these solutions are unstable
\cite{NorgardG:08a,BhatHS:08a}.

\item[5.]  It is conjectured that solutions to Equation (\ref{CFB}) will converge to the entropy solution of inviscid Burgers
Equation for all continuous initial conditions, and thus the scheme
\begin{eqnarray}
u_t +\ubar u_x=0\\
\ubar=g^\alpha \ast u\\
u(x,0)=g^\alpha \ast u_0(x),
\end{eqnarray}
will converge to the entropy solution for any bounded initial
condition $u_0$ \cite{NorgardG:08a}.
\end{itemize}

The work on the regularization of the Burgers equation is inspired
by and related to work done on the LANS-$\alpha$ equations
\cite{Holm:04a,Foias:01a,Marsden:01h,Mohseni:03a,Mohseni:05e,Marsden:98b,Chen:98a}.
These equations also employ an averaged velocity in the nonlinear
term and have been successful in modeling some turbulent
incompressible flows.

It is thought that a similar regularization could be accomplished
for the equations that describe compressible flow. Encouraged by the
results for Burgers equation, the next step is to attempt to
introduce averaging into the one-dimensional homentropic Euler
equations, a simplified version of the full Euler equations, where
pressure is purely a function of density. There have been several
attempts at such a regularization.

Using a Lagrangian averaging technique Bhat and Fetecau
\cite{BhatHS:06b} derived the following equations
\begin{eqnarray}
\rho_t+(\rho u)_x=0\\
w_t+(uw)_x-\frac{1}{2}( u^2 + \alpha^2 u_x^2)_x=-\frac{p_x}{\rho}\\
\rho w=\rho v-\alpha^2 \rho_x u_x\\
v=u-\alpha^2 u_{xx}.
\end{eqnarray}
While the solutions to the system remained smooth and contained much
structure it was found that the equations were `` not well-suited
for the approximation of shock solutions of the compressible Euler
equations.''

Another attempt by Bhat, Fetecau, and Goodman used a Leray-type
averaging  \cite{BhatHS:07a} leading to the equations
\begin{eqnarray}\label{LerayAveragedHomentropicEuler}
\rho_t+\ubar \rho_x +\rho u_x=0\\
u_t + \ubar u_x + \frac{p_x}{\rho}=0\\
u=\ubar-\alpha^2 \ubar_{xx}.
\end{eqnarray}
with $p=\kappa \rho ^\gamma$.  They then showed that weakly
nonlinear geometrical optics (WNGO) asymptotic theory predicts the
equations will have  global smooth solutions for $\gamma =1$ and
form shocks in finite time for $\gamma \neq 1$.

Additionally in 2005, H. S. Bhat et. al. \cite{Mohseni:05d} applied
the Lagrangian averaging approach to the full compressible Euler
equations.  Their approach was successful in that a set of
Lagrangian Averaged Euler (LAE-$\alpha$) equations were derived.
However, the equations seemed so intractable that they seemed
impractical for real world applications.

Inspired by the existence uniqueness proofs from the averaged
Burgers equations found in \cite{Norgard:08b,BhatHS:06a}, we average
the characteristics of the homentropic Euler equations to derive
what we term as the characteristically averaged homentropic Euler
(CAHE) equations \eref{CAHE}.  It is the properties of these
equations that this paper examines.

The following section follows the derivation and present the final
equations.  The existence and uniqueness of solutions to the
equations are then proven in section \ref{ExistenceTheoremSection}.
Sections \ref{shockspeedsection} and \ref{ReimannSolutionsSection}
examine the speed of the shocks and solutions to the Riemann
problem.  Numerical simulations and their comparison to those of the
homentropic Euler equations are discussed in sections
\ref{Numericssection} and \ref{Comparisonsection}.  The results are
then briefly summarized in the concluding remarks.

\section{Derivation of the equations}\label{sectionderivation}
\subsection{The homentropic Euler equations}
We begin the process of deriving the CAHE equations by starting with
the homentropic Euler equations.  There are two equations to the
homentropic Euler equations.  Conservation of mass and conservation
of momentum.  Pressure is expressed purely density raised to the
power of $\gamma$.
\begin{subequations}
\label{HomentropicEulerEquations}
\begin{eqnarray}
\label{HomentropicEulerEquationsa}
\rho_t+ (\rho u)_x=0\\
\label{HomentropicEulerEquationsb}
(\rho u)_t+ (\rho u u+\rho^\gamma)_x=0
\end{eqnarray}
\end{subequations}
The equations are then written in primitive variable form
\begin{equation}\label{HomentropicEulerEquationsVector}
\left[
  \begin{array}{c}
    \rho \\
    u \\
  \end{array}
\right]_t +\left[
   \begin{array}{cc}
     u & \rho  \\
     \frac{a^2}{\rho} & u \\
   \end{array}
 \right]
 \left[
  \begin{array}{c}
    \rho \\
    u \\
  \end{array}
\right]_x =0,
\end{equation}
with $a^2=\gamma \rho^{\gamma-1}.$

 To get a clear view of the characteristics involved the equations are then
 diagonalized to obtain the equations
\begin{equation} \label{HomentropicEulerEquationsDiagonalized}
\left[
  \begin{array}{c}
    v^+ \\
    v^- \\
  \end{array}
\right]_t +\left[
   \begin{array}{cc}
     u +a & 0 \\
     0 & u -a \\
   \end{array}
 \right]
 \left[
  \begin{array}{c}
    v^+ \\
    v^- \\
  \end{array}
\right]_x =0
\end{equation}
where
\begin{subequations}
\begin{eqnarray}
v^+=u+\frac{2a}{\gamma-1}\\
v^-=u-\frac{2a}{\gamma-1}.
\end{eqnarray}
\end{subequations}
The variables $v^\pm$ are commonly know as Riemann invariants. From
Equation \eref{HomentropicEulerEquationsDiagonalized} it is easy to
see that the quantity $v^+$ is convected at speed $u+a$ and $v^-$
speed $u-a$.  Thus along the characteristics $u+a$, $v^+$ will
remain constant.

Shocks will form when characteristics intersect.  To prevent this
from happening the characteristics are spatially averaged. This
averaging is conducted by convoluting the variable to be averaged
with an averaging kernel, $g$, and is represented by a bar above the
variable. For example, the averaged velocity would be expressed
\begin{equation}\ubar=g \ast u .\end{equation}

\subsection{The averaging kernel}
Several assumptions on the averaging kernel are made at this point.
The kernel is assumed to be even, for isotropic purposes.  For the
purpose of Theorems \ref{existencetheorem} and
\ref{existencetheorem2} the kernel and its first derivative are
assumed to be integrable. Of special interest is the Helmholtz
filter which is defined as
\begin{equation}
u=\ubar-\ubar_{xx},
\end{equation}
and thus has an averaging kernel of
\begin{equation}
g(x)=\frac{1}{2}\exp(-|x|).
\end{equation}
For all numerical simulations found in Sections
\ref{ReimannSolutionsSection}, \ref{Numericssection}, and
\ref{Comparisonsection} the Helmholtz filter is used for its
convenient inversion techniques.

Furthermore the kernel will be equipped with a parameter $\alpha$
which will control the amount of averaging.  If $g$ is the averaging
kernel, then $\alpha$ is introduced as
\begin{equation}
g^\alpha=\frac{1}{\alpha} g\left( \frac{x}{\alpha} \right).
\end{equation}
Thus as $\alpha \to 0$ the averaging kernel approaches the Dirac
delta function.

\subsection{Obtaining the CAHE equations}
Using the averaging kernel discussed above the characteristics of
the homentropic Euler equations are averaged to obtain the CAHE
equations in diagonalized form.
\begin{equation} \label{CAHEdiagonalized}
\left[
  \begin{array}{c}
    v^+ \\
    v^- \\
  \end{array}
\right]_t +\left[
   \begin{array}{cc}
     \ubar +\abar & 0 \\
     0 & \ubar -\abar \\
   \end{array}
 \right]
 \left[
  \begin{array}{c}
    v^+ \\
    v^- \\
  \end{array}
\right]_x =0.
\end{equation}
Change the equations back into the primitive variable form and we
get the equations
\begin{subequations}
\begin{eqnarray}
\rho_t+\ubar \rho_x +\rho \frac{\abar}{a} u_x=0 \nonumber \\
u_t+ \ubar u_x+\frac{a\abar}{\rho}\rho_x=0 \nonumber \\
\ubar=g \ast u \nonumber \\
\abar=g \ast a \nonumber
\end{eqnarray}
\end{subequations}
with $a^2=\gamma \rho^{\gamma-1}$ which were seen earlier as
Equations \eref{CAHE}.

These are now the equations that will be examined for the rest of
this paper and are referred to as the characteristically averaged
homentropic Euler (CAHE) equations. We begin by proving that
Equations (\ref{CAHE}) have one and only one solution.

\section{Existence and uniqueness theorem}\label{ExistenceTheoremSection}

A critically important property for the Equations (\ref{CAHE}) to
have is that a solution exists. This section addresses this problem
by presenting a proof for the existence and uniqueness of solutions.
This proof uses a method of characteristics approach and is similar
to the existence uniqueness proof of Convectively Filtered Burgers
equation found in \cite{Norgard:08b}

An outline of the proof is as follows.  When the equations are cast
into their characteristic form, it is clear to see that the fluid
velocity and the speed of sound remain bounded. Since those speeds
remain bounded, the first derivative of the averaged speeds will
remain bounded as well. The characteristics are governed by the
averaged speeds and with the first derivatives bounded, the
characteristics may grow closer, but will never intersect.  Thus the
solution can be fully realized by the characteristics and the
initial conditions.

\begin{theorem}
\label{existencetheorem}
Let $g(\mathbf{x})$ $\in$ $W^{1,1}(\mathbb{R})$ and $u_0(x), a_0(x)$ $\in$
$C^1(\mathbb{R})$, then there exists a unique global solution
$u(x,t), a(x,t) \in C^1(\mathbb{R},\mathbb{R})$ to the initial value problem

\begin{subequations}
\label{isentropiceuler}
\begin{eqnarray}
\label{isentropiceulera}
\rho_t+\ubar \rho_x +\rho \frac{\abar}{a} u_x=0\\
\label{isentropiceulerb}
u_t+ \ubar u_x+\frac{a\abar}{\rho}\rho_x=0\\
\label{isentropiceulerc}
\ubar=g \ast u\\
\label{isentropiceulerd}
\abar=g \ast a\\
u(x,0)=u_0\\
a(x,0)=a_0
\end{eqnarray}
\end{subequations}
with $a^2=\gamma \rho^{\gamma-1}.$
\end{theorem}

\begin{proof} Express Equations (\ref{isentropiceulera}) and (\ref{isentropiceulerb}) with matrices as

\begin{equation}
\left[
  \begin{array}{c}
    \rho \\
    u \\
  \end{array}
\right]_t
+\left[
   \begin{array}{cc}
     \ubar & \rho \frac{\abar}{a} \\
     \frac{a\abar}{\rho} & \ubar \\
   \end{array}
 \right]
 \left[
  \begin{array}{c}
    \rho \\
    u \\
  \end{array}
\right]_x =0
\end{equation}

By diagonalizing the matrix these equations can be rewritten as
\begin{equation} \label{diagonalizedisentropiceuler}
\left[
  \begin{array}{c}
    v^+ \\
    v^- \\
  \end{array}
\right]_t
+\left[
   \begin{array}{cc}
     \ubar +\abar & 0 \\
     0 & \ubar -\abar \\
   \end{array}
 \right]
 \left[
  \begin{array}{c}
    v^+ \\
    v^- \\
  \end{array}
\right]_x =0.
\end{equation}

It is here that we shift our perspective to a method of
characteristics type view. Associate the maps $\phi^\pm$ with the
characteristics of $v^\pm$. Thus we have
\begin{eqnarray}
\label{positivecharacteristic}
\frac{\partial}{\partial t} \phi^+(\xi,t) &=&\ubar(\phi^+(\xi,t),t)+\abar(\phi^+(\xi,t),t)\\
\frac{\partial}{\partial t} v^+&=&0 \qquad \text{along a $\phi^+$
characteristic }
\end{eqnarray}
and
\begin{eqnarray}
\label{negativecharacteristic}
\frac{\partial}{\partial t} \phi^-(\xi,t) &=& \ubar(\phi^-(\xi,t),t)-\abar(\phi^-(\xi,t),t)\\
\frac{\partial}{\partial t} v^-&=&0 \qquad \text{along a $\phi^-$
characteristic }.
\end{eqnarray}

From this we can obtain the estimates that
\begin{equation}
\label{linfinitybound}
\left|\left| v^\pm(x,t) \right|\right|_{L^\infty}=\left|\left| v^\pm(x,0) \right|\right|_{L^\infty}.
\end{equation}

It the mappings $\phi^\pm$ have continuously differentiable
inverses, $\varphi^\pm$, then Equation
(\ref{diagonalizedisentropiceuler}) has the solution
\begin{equation}
\label{solution}
v^\pm(x,t)=v_0^{\pm}(\varphi^\pm(x,t)).
\end{equation}

Sufficient conditions for such inverses to uniquely exist is if the
Jacobians of $\phi^\pm$ are non-zero for all positions and time.
Thus if $J(\phi^\pm) \neq 0$, Equation
(\ref{diagonalizedisentropiceuler}) is uniquely solved by
(\ref{solution}).

For the next section of the proof we will be dealing with $\phi^+$.
The results for $\phi^-$ follow in precisely the same manner. Since
we are dealing with 1D the Jacobian of $\phi^+$ is essentially
$\phi^+_x$.  It is clear from Equation
(\ref{positivecharacteristic}) to see that the time derivative of
$\phi^+_x$ is
\begin{equation}
\frac{\partial}{\partial t}  \phi^+_x= (\ubar_x +\abar_x)  \phi^+_x.
\end{equation}
Thus we see that
\begin{equation}
\phi^+_x=\phi^+_x(0)\exp \left( \int_0^t \ubar_x+\abar_x \,dt \right).
\end{equation}
Thus $\phi^+_x$ will remain non-zero if $\left|\int_0^t \ubar_x+\abar_x \,dt\right| <  \infty$.

First we will show that $u,a \in L^\infty$.  We see that
\begin{eqnarray}
\left|\left| 2u  \right|\right|_{L^\infty}&=& \left|\left| \frac{2a}{\gamma-1}+u + u - \frac{2a}{\gamma-1}  \right|\right|_{L^\infty}\\
&\leq& \left|\left| v^+  \right|\right|_{L^\infty} + \left|\left| v^-  \right|\right|_{L^\infty}.
\end{eqnarray}
Looking at Equation (\ref{linfinitybound}) and noting that
$\left|\left| v^+  \right|\right|_{L^\infty}$ and $\left|\left| v^-
\right|\right|_{L^\infty}$
 are bounded for all time, then $\left|\left| u  \right|\right|_{L^\infty}$ is bounded for all time.
 Similarly one can bound the quantity $\left|\left| a \right|\right|_{L^\infty}$ for all time.

Given that $g(x)$ $\in$ $W^{1,1}(\mathbb{R})$, there exists $M \in
\mathbb{R}$, such that
\begin{equation}
\label{l1bound}
\left|\left| \frac{\partial}{\partial{x} } g \right|\right|_{L^1} \leq M<\infty.
\end{equation}
Knowing that $\frac{\partial}{\partial{x} } g \in L^1$ and $u \in L^\infty$,
we know that $\frac{\partial}{\partial{x} } \ubar$ exists and that
\begin{equation*}
\frac{\partial}{\partial{x} } \ubar =\frac{\partial}{\partial{x} } g \ast u.
\end{equation*}
Similarly
\begin{equation*}
\frac{\partial}{\partial{x} } \abar =\frac{\partial}{\partial{x} } g \ast a.
\end{equation*}

Using Young's inequality we can bound the derivatives $\ubar_x$ and $\abar_x$.  We get
\begin{eqnarray}
\left|\left| \frac{\partial}{\partial{x }} \ubar \right|\right|_{L^\infty}
\leq \left|\left| \frac{\partial}{\partial{x} } g \right|\right|_{L^1}
\left|\left| u \right|\right|_{L^\infty} \leq M \left|\left| u \right|\right|_{L^\infty} \\
\left|\left| \frac{\partial}{\partial{x }} \abar \right|\right|_{L^\infty}
\leq \left|\left| \frac{\partial}{\partial{x} } g \right|\right|_{L^1}
\left|\left| a \right|\right|_{L^\infty} \leq M \left|\left| a \right|\right|_{L^\infty}.
\end{eqnarray}
This leads directly to the bound
\begin{equation}
\left|\int_0^t \ubar_x+\abar_x \,dt\right| <  M \left ( \left|\left|
u \right|\right|_{L^\infty}+\left|\left| a \right|\right|_{L^\infty}
\right) t
\end{equation}

Thus for finite time, the Jacobian of $\phi^\pm$ remains uniquely
invertible, with a continuously differentiable inverse and thus
(\ref{solution}) is a unique $C^1(\mathbb{R}^n)$ solution to
(\ref{diagonalizedisentropiceuler}).  A unique  $C^1(\mathbb{R}^n)$
solution to (\ref{isentropiceuler}) follows accordingly.
\end{proof}

\begin{theorem}
\label{existencetheorem2} Let $g(\mathbf{x})$ $\in$
$W^{1,1}(\mathbb{R})$ and $u_0(x), a_0(x)$ $\in$
$L^\infty(\mathbb{R})$, then there exists a unique global solution
$u(x,t), a(x,t) \in L^\infty(\mathbb{R},\mathbb{R})$ to the initial
value problem

\begin{subequations}
\begin{eqnarray}
\rho_t+\ubar \rho_x +\rho \frac{\abar}{a} u_x=0\\
u_t+ \ubar u_x+\frac{a\abar}{\rho}\rho_x=0\\
\ubar=g \ast u\\
\abar=g \ast a\\
u(x,0)=u_0\\
a(x,0)=a_0
\end{eqnarray}
\end{subequations}

\end{theorem}

\begin{proof}
The proof is the same as for Theorem \ref{existencetheorem}.
$\phi^\pm$ still have unique continuously differentiable inverses,
and the solution remains in the same form as Equation
(\ref{solution}), but now lacks continuity due to the initial
conditions.
\end{proof}

It should be noted that different averagings can be used for $\ubar$
and $\abar$. The existence and uniqueness proof will hold as long as
a bound on the first derivatives of $\ubar$ and $\abar$ remains
constant. This is noted as the type of averaging used will affect
shock speed as noted in Section \ref{shockspeedsection}.

\section{Shock speeds}\label{shockspeedsection}
One of the consequences of the previous sections is that for
continuous initial conditions, the solution will remain continuous.
Also for initial conditions with a discontinuity, that discontinuity
will remain. There are no colliding characteristics.  Thus the
``shock'' is really a discontinuity that is being convected.  If
there is a discontinuity in $v^+_0(x)$ then that discontinuity will
travel at speed $\ubar(x^*) +\abar(x^*)$, where $x^*$ is the
location of the discontinuity.  Similarly, if there is a
discontinuity in $v^-_0(x)$ then that discontinuity will travel at
speed $\ubar(x^*) -\abar(x^*)$.

Consider the case where there is a single jump discontinuity in
$v^+$, but otherwise constant.
\begin{subequations}\label{singleshock}
\begin{eqnarray}
v^+(x)=\left\{\begin{array}{ll}
v^+_l & x<0 \\
v^+_r & x\geq 0
\end{array}\right.\\
v^-(x)=C.
\end{eqnarray}
\end{subequations}
This will lead to a jump discontinuity in $u$ and in $a$.

With $\ubar$ and $\abar$ defined as in Equations
(\ref{isentropiceuler}) and the filter $g$ being even and $\int g
=1$ then the speed of the discontinuity in Equations
(\ref{singleshock}) will be
\begin{equation}
s=\frac{u_l+u_r}{2} + \frac{a_l+a_r}{2}
\end{equation}
where $u_l,u_r,a_l$ and $a_r$ are the limiting values of $u$ and $a$
on the left and right side of the discontinuity respectively.

Similarly, if the jump discontinuity existed only in $v^-$ then the
speed of the discontinuity would be
\begin{equation}
s=\frac{u_l+u_r}{2} - \frac{a_l+a_r}{2}.
\end{equation}

For multiple discontinuities or for nonconstant values around the
discontinuities, the values of $\ubar$ and $\abar$ cannot generally
be analytically written . However, for small values of $\alpha$ the
filter will not ``see'' as far and the values of $\ubar$ and $\abar$
will be roughly the average of the values just to the left and right
of the discontinuity.

\subsection{Alternative averagings and shock speeds}

In the above calculation $\ubar$ and $\abar$ were defined  using the
definition
\begin{eqnarray}
\bar{(\cdot)}=g \ast (\cdot)
\end{eqnarray}
This is not, however, the only way to conduct averaging. One
possible alternative is using Favre averaging, a density weighted
average, where
\begin{eqnarray}\label{altbar}
\bar{(\cdot)}=\frac{g \ast (\rho\,\, \cdot)}{g \ast (\rho)}.
\end{eqnarray}

This alternative will not affect the existence uniqueness theorems
established in section \ref{ExistenceTheoremSection}.  It will,
however, change the shock speeds established earlier.

If Equation (\ref{altbar}) is used, the speed of a discontinuity in
$v^+$ will travel at
\begin{equation}
s=\frac{\rho_l u_l+\rho_r u_r}{\rho_l +\rho_r } + \frac{\rho_l
a_l+\rho_r a_r}{\rho_l +\rho_r }.
\end{equation}

Another possibility is to apply a spatial average to the conserved
quantities, density and momentum, and then compute the
characteristic speeds from the the averaged conserved quantities,
\begin{eqnarray}
\tilde{u}=\frac{\overline{\rho u}}{\bar{\rho}}\\
\tilde{a}^2=\gamma (\bar{\rho})^{\gamma-1}.
\end{eqnarray}
If the speed of the characteristics were defined using $\tilde{u}$
and $\tilde{a}$, then the speed of the shocks would be
\begin{equation}
s=\frac{\rho_l u_l+\rho_r u_r}{\rho_l +\rho_r } + \sqrt{\gamma}
\left( \frac{\rho_l+\rho_r}{2} \right)^{ \frac{\gamma-1}{2}}.
\end{equation}
Thus how one chooses to define $\ubar$ and $\abar$ will affect the
speed of the discontinuities.

\section{Riemann solutions}\label{ReimannSolutionsSection}
This section builds upon the results of section
\ref{shockspeedsection} and establishes properties of the solutions
of the CAHE equations for the Riemann problem.  It then continues on
to examine solutions of the CAHE equations if, instead of the
Riemann problem, the initial conditions are slightly perturbed.

\subsection{Riemann problem}

Looking at Equation (\ref{solution}), it is clear that any
discontinuity that exists in the initial condition remains in the
solutions for all time.  Furthermore, no additional discontinuities
will form.

For the Riemann problem, at time $t=0$ a single discontinuity will
exist at the origin:
\begin{eqnarray}
u(x,0)=\left\{\begin{array}{ll}
u_l & x<0 \\
u_r & x\geq 0
\end{array}\right.\\
a(x,0)=\left\{\begin{array}{ll}
a_l & x<0 \\
a_r & x\geq 0
\end{array}\right.
\end{eqnarray}
which can also be rewritten in terms of the variables $v^+$ and
$v^-$.
\begin{eqnarray}
v^+(x,0)=\left\{\begin{array}{ll}
v^+_l & x<0 \\
v^+_r & x\geq 0
\end{array}\right.\\
v^-(x,0)=\left\{\begin{array}{ll}
v^-_l & x<0 \\
v^-_r & x\geq 0
\end{array}\right.
\end{eqnarray}

These discontinuities will travel at the speeds discussed in section
\ref{shockspeedsection}.  We denoted the locations of the
discontinuities in $v^+$ and $v^-$ as $x^-$ and $x^+$. The speeds of
these discontinuities are determined by
\begin{equation}\label{discontinuityspeed}
\frac{\p}{\p t} x^\pm = \ubar(x^\pm) \pm \abar (x^\pm).
\end{equation}
Clearly the discontinuity found in $v^+$ will travel at a faster
speed than the discontinuity in $v^-$.  Thus the solution will
consist of three different areas separated by the locations of $x^-$
and $x^+$.  The solution is then
\begin{eqnarray}
v^+(x,t)=\left\{\begin{array}{ll}
v^+_l & x<x^+(t) \\
v^+_r & x^+(t)\leq x
\end{array}\right.\\
v^-(x,t)=\left\{\begin{array}{ll}
v^-_l & x<x^-(t) \\
v^-_r & x^-(t)\leq x
\end{array}\right.
\end{eqnarray}
If expressed in primitive variables the solution is
\begin{subequations}\label{reimannsolution}
\begin{eqnarray}
u(x,0)=\left\{\begin{array}{ll}
u_l & x<x^-(t) \\
u_m & x^-(t)\leq  x < x^+(t) \\
u_r &  x^+(t) \leq  x
\end{array}\right.\\
a(x,0)=\left\{\begin{array}{ll}
a_l & x<x^-(t) \\
a_m & x^-(t)\leq x<x^+(t) \\
a_r & x^+(t)\leq x
\end{array}\right.
\end{eqnarray}
\end{subequations}
where $u_m=\frac{v^+_l +v^-_r}{2}$ and $a_m=\frac{\gamma-1}{2}
\frac{v^+_l -v^-_r}{2}$.

Thus in general, the Riemann problem will produce the solutions
presented in Equations (\ref{reimannsolution}) where two
discontinuities will propagate dependent upon the speeds of the
averaged velocity and speed of sound.  If the initial conditions
happen to be chosen, such that $v^+$ or $v^-$ are constant, then
there will exist only one traveling discontinuity.  The following
subsection show numerical simulations of a typical Riemann solution.

\subsubsection{Riemann solution numerics}\label{RiemannNumericsSection}
In section \ref{Numericssection}, a numerical scheme is described
that numerically simulates the behavior of the CAHE equations.
However in section \ref{peturbedRiemannSection} it is demonstrated
that some of the solutions to the Riemann problem are unstable and
thus cannot be capture by the numerical scheme.  Thus this numerical
scheme was developed.  Since the Riemann problem solutions are of
the form presented in Equations (\ref{reimannsolution}), finding
solutions, even unstable ones, to the Riemann problem reduces to
tracking the locations of the discontinuities in time. The following
numerical scheme does precisely this.

To find the solution (\ref{reimannsolution}), the values $u_{l,m,r}$
and $a_{l,m,r}$ are all known, one simply has to determine the
location of $x^\pm$.  This is done by using Equation
(\ref{discontinuityspeed}).  A typical iteration is as follows.

\begin{itemize}
\item[1.] The values of $u(x,t_1)$ and $a(x,t_1)$ are known on a uniform grid.
\item[2.] The Helmholtz operator is then numerically inverted to give the values of $\ubar$ and $\abar$ on the uniform grid.  These values are then interpolated to give the values of $\ubar \pm \abar$ at $x^\pm$.
\item[3.] The positions of the $x^\pm$, are then advanced in time by the values of $\ubar \pm \abar$ on $x^\pm$ calculated in the previous step.
\item[4.] The values of $u(x,t_2)$ and $a(x,t_2)$ are now known using Equation (\ref{reimannsolution}) and current positions of $x^\pm$.
\end{itemize}
The inversion of the Helmholtz operator is described in section
\ref{Numericssection}.

To examine a typical case, we will simulate the Riemann solution with the initial conditions
\begin{subequations} \label{Example1}
\begin{eqnarray}
u(x,0)=\left\{\begin{array}{ll}
0 & x<0 \\
0 & x\geq 0
\end{array}\right.\\
a(x,0)=\left\{\begin{array}{ll}
2 & x<0 \\
1 & x\geq 0
\end{array}\right.
\end{eqnarray}
\end{subequations}
This set of initial conditions will be referred to as Example 1a.
This is a standard shock tube problem where the velocity is zero
initially and there is a jump in pressure. The initial conditions
can be seen in Figure \ref{riemannproblemfigure1}. With these values
there are discontinuities in variables $v^+$ and $v^-$. These
discontinuities are seen to propagate to the right and left
respectively as seen in Figure \ref{riemannproblemfigure2}.  These
numerical simulations were conducted with $2^{10}$ grid points on
the domain $[-1,1]$ with $\alpha=0.02$.

\begin{figure}[!ht]
\begin{center}
\begin{minipage}{0.48\linewidth} \begin{center}
  \includegraphics[width=.9\linewidth]{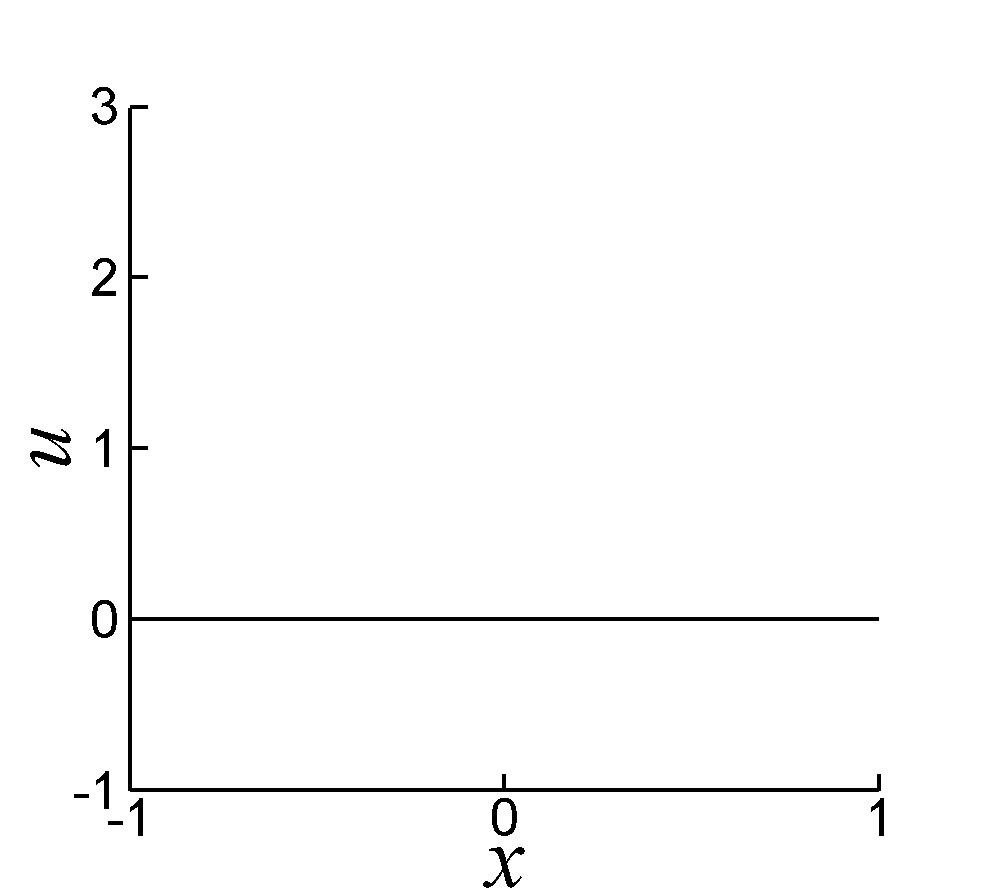}
\end{center} \end{minipage}
\begin{minipage}{0.48\linewidth} \begin{center}
  \includegraphics[width=.9\linewidth]{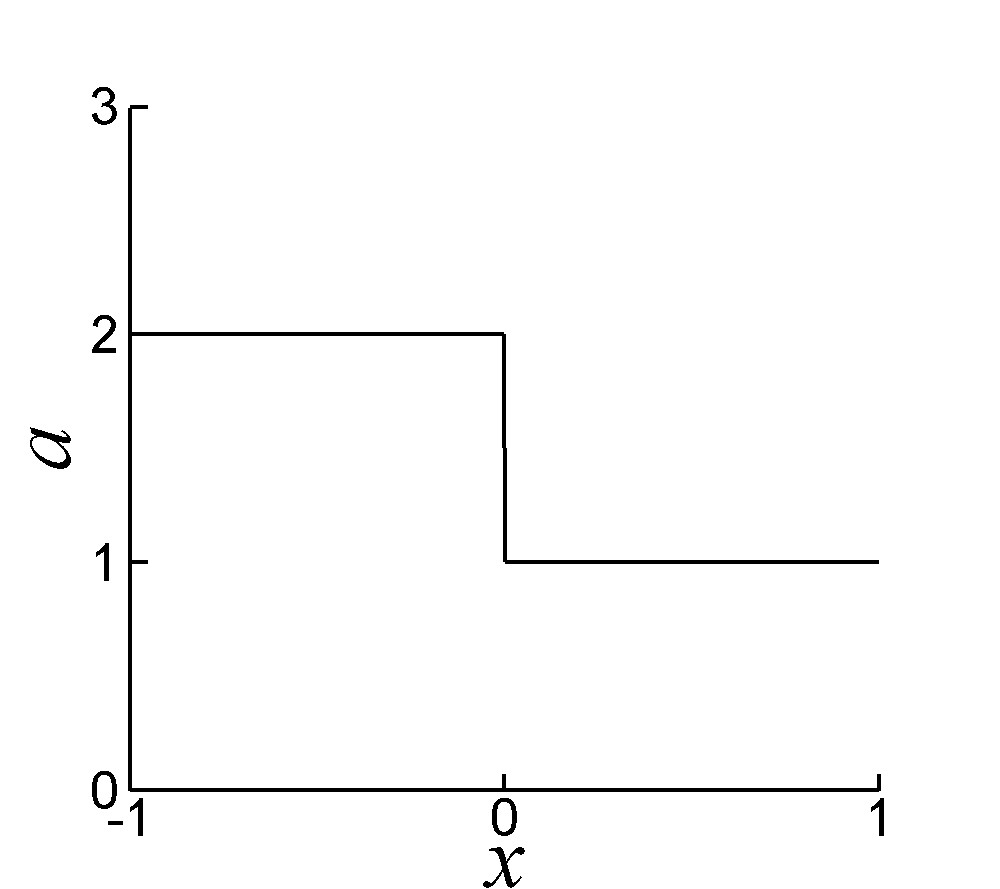}
\end{center} \end{minipage}\\
\begin{minipage}{0.48\linewidth}\begin{center} (a) \end{center} \end{minipage}
\begin{minipage}{0.48\linewidth}\begin{center} (b) \end{center}
\end{minipage}\vspace{-2mm}
\caption{The initial conditions for the Riemann problem. (a)  The
initial velocity profile is $u=0$. (b) The initial speed of sound
profile is a simple jump discontinuity, indicating a higher pressure
on the right. } \label{riemannproblemfigure1}
\end{center}
\end{figure}

\begin{figure}[!ht]
\begin{center}
\begin{minipage}{0.48\linewidth} \begin{center}
  \includegraphics[width=.9\linewidth]{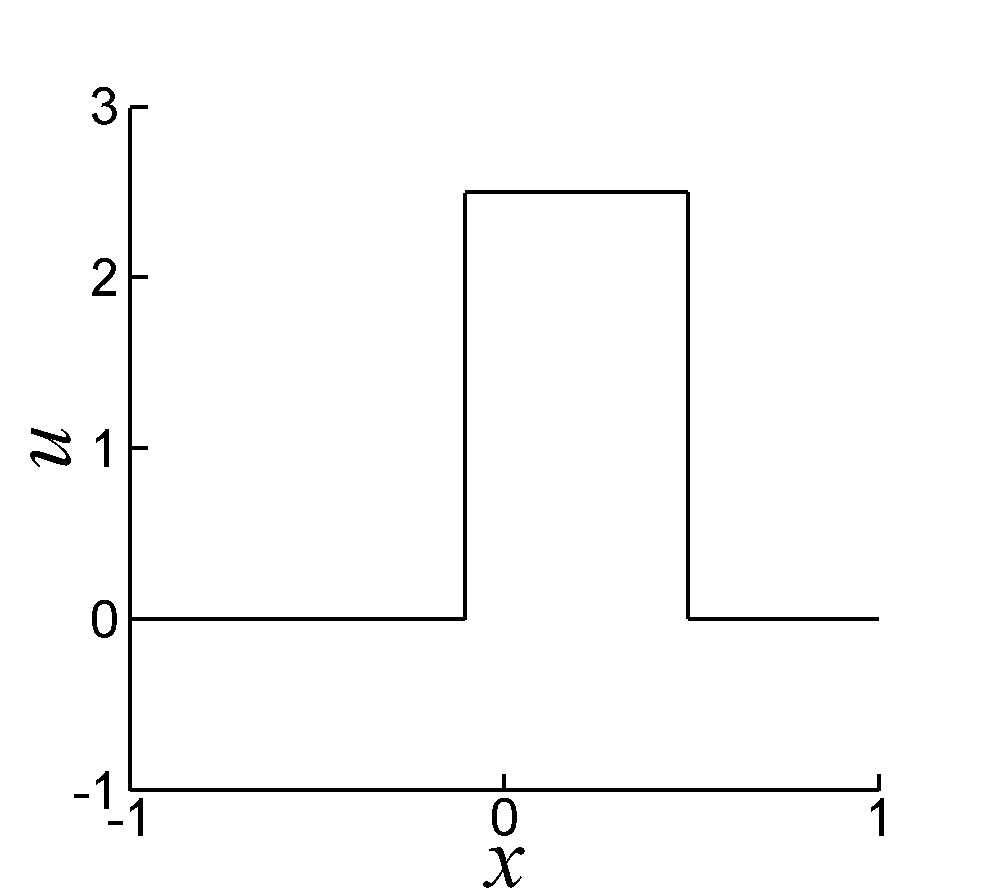}
\end{center} \end{minipage}
\begin{minipage}{0.48\linewidth} \begin{center}
  \includegraphics[width=.9\linewidth]{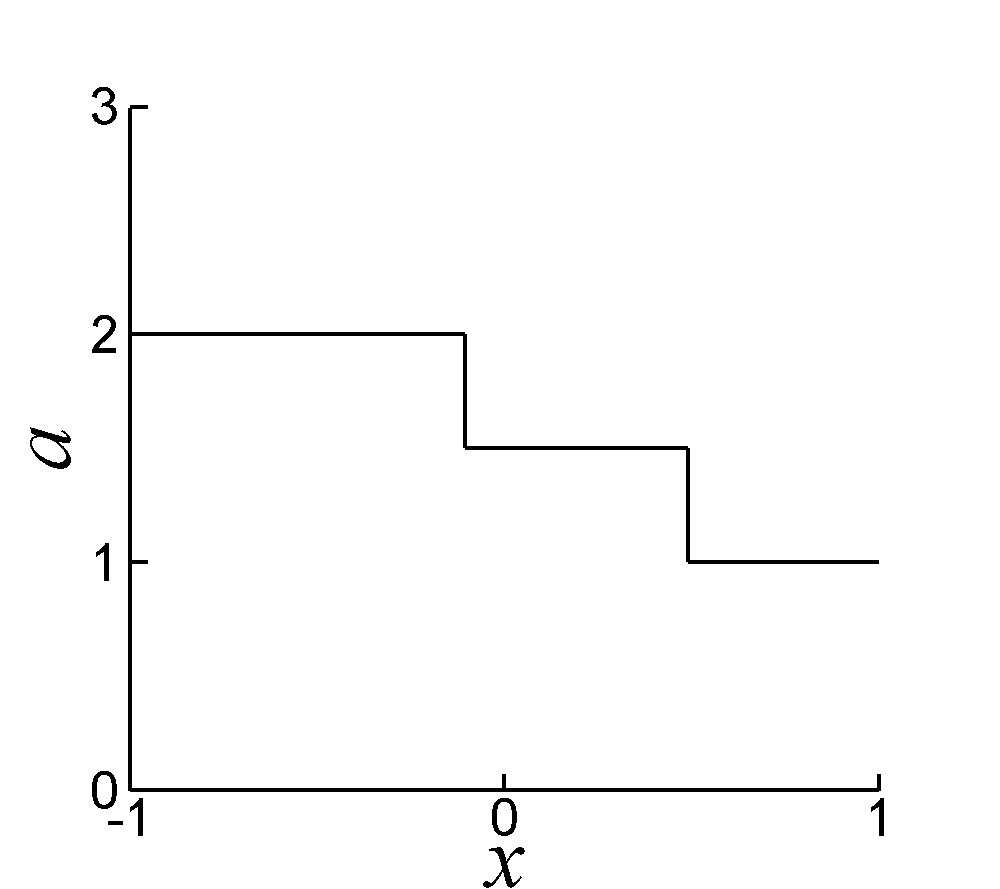}
\end{center} \end{minipage}\\
\begin{minipage}{0.48\linewidth}\begin{center} (a) \end{center} \end{minipage}
\begin{minipage}{0.48\linewidth}\begin{center} (b) \end{center}
\end{minipage}\vspace{-2mm}
\caption{The solution to the Riemann problem with initial conditions
Equation (\ref{Example1}).   There are two shocks found in both the
velocity and speed of sound.  This is the solution at time $t=0.2.$
} \label{riemannproblemfigure2}
\end{center}
\end{figure}

\subsection{The perturbed Riemann problem}\label{peturbedRiemannSection}
While the solution to the
Riemann problem generally leads to one or two traveling
discontinuities, we have found that some of these traveling
discontinuities are unstable.

There have been examinations of the Burgers equation with the
characteristics averaged \cite{Mohseni:06l, Norgard:08b,
NorgardG:08a, BhatHS:06a, BhatHS:08a}.  In two of these papers it
was seen that certain traveling discontinuities in these equations
were unstable and if perturbed would become expansion waves
\cite{NorgardG:08a, BhatHS:08a}. We have found similar behavior for
the CAHE equations.

Consider Example 1a as seen in Figure \ref{riemannproblemfigure2}.
The position of the leftmost shock is denoted $x^-$.  The value of
$u-a$ is greater to the right of $x^-$ than it is to the left.  So
clearly $\ubar - \abar$ is increasing across the discontinuity at
$x^-$.  Thus the $v^-$ characteristics are diverging there.  This is
typically indicative of an expansion wave.  However, the spreading
of this discontinuity remains simply a discontinuity.


 Suppose instead that instead of a strict discontinuity, the function was a
continuous function, albeit extremely steep.  The diverging
characteristics would then spread this extremely steep gradient and
gradually make it less steep, forming an expansion wave.  Thus, this
traveling discontinuity is unstable, and any amount of smoothing
would lead to an expansion wave as opposed to a shock.

Using numerics discussed in section \ref{Numericssection}, we made
runs using the same initial conditions as in Example 1a except the
initial conditions were initially ``smoothed'', using the same
Helmholtz filter as employed in the averaging of the
characteristics.  We will refer to this as Example 1b.  This
smoothing of the initial conditions is similar to the approach used
in \cite{NorgardG:08a}. Figure \ref{riemannproblemfigure3} shows the
initial conditions where the smoothing is readily apparent in the
$a$ variable. Figure \ref{riemannproblemfigure4} shows the solution
at time $t=0.2$.  The right shock remains unchanged from the
unperturbed solution seen in Figure \ref{riemannproblemfigure2}.
However, an expansion wave is clearly seen on the left in the
perturbed solution where previously there was a discontinuity. These
simulations were run with a resolution of $2^{10}$ grid points and
$\alpha=0.02$.

\begin{figure}[!ht]
\begin{center}
\begin{minipage}{0.48\linewidth} \begin{center}
  \includegraphics[width=.9\linewidth]{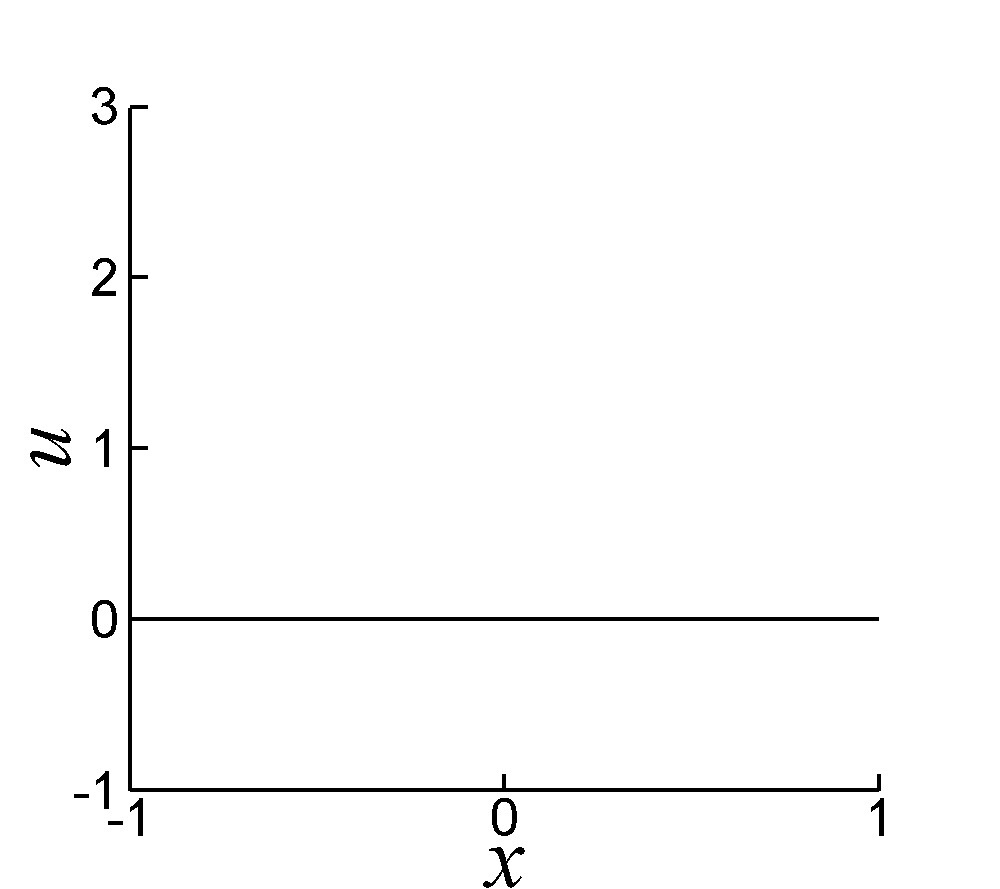}
\end{center} \end{minipage}
\begin{minipage}{0.48\linewidth} \begin{center}
  \includegraphics[width=.9\linewidth]{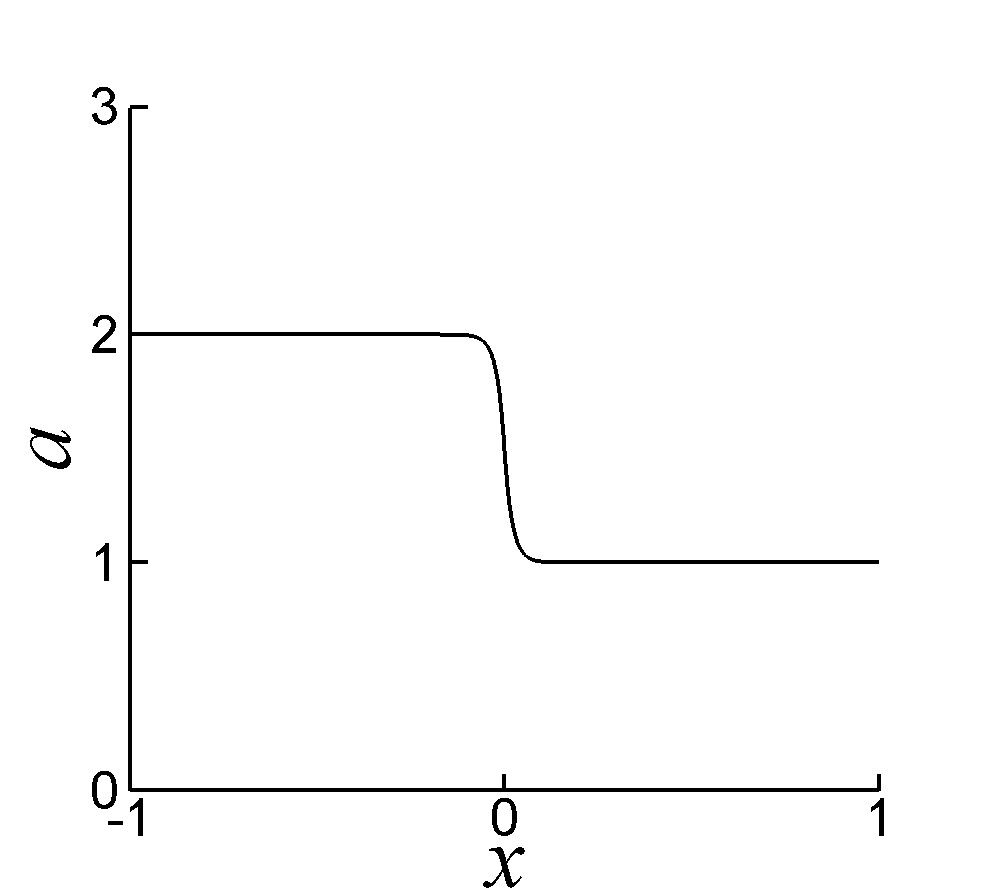}
\end{center} \end{minipage}\\
\begin{minipage}{0.48\linewidth}\begin{center} (a) \end{center} \end{minipage}
\begin{minipage}{0.48\linewidth}\begin{center} (b) \end{center}
\end{minipage}\vspace{-2mm}
\caption{The initial conditions for the perturbed Riemann problem.
(a) The initial velocity profile is $u=0$. (b) The initial speed of
sound profile is a smoothed pressure jump. }
\label{riemannproblemfigure3}
\end{center}
\end{figure}

\begin{figure}[!ht]
\begin{center}
\begin{minipage}{0.48\linewidth} \begin{center}
  \includegraphics[width=.9\linewidth]{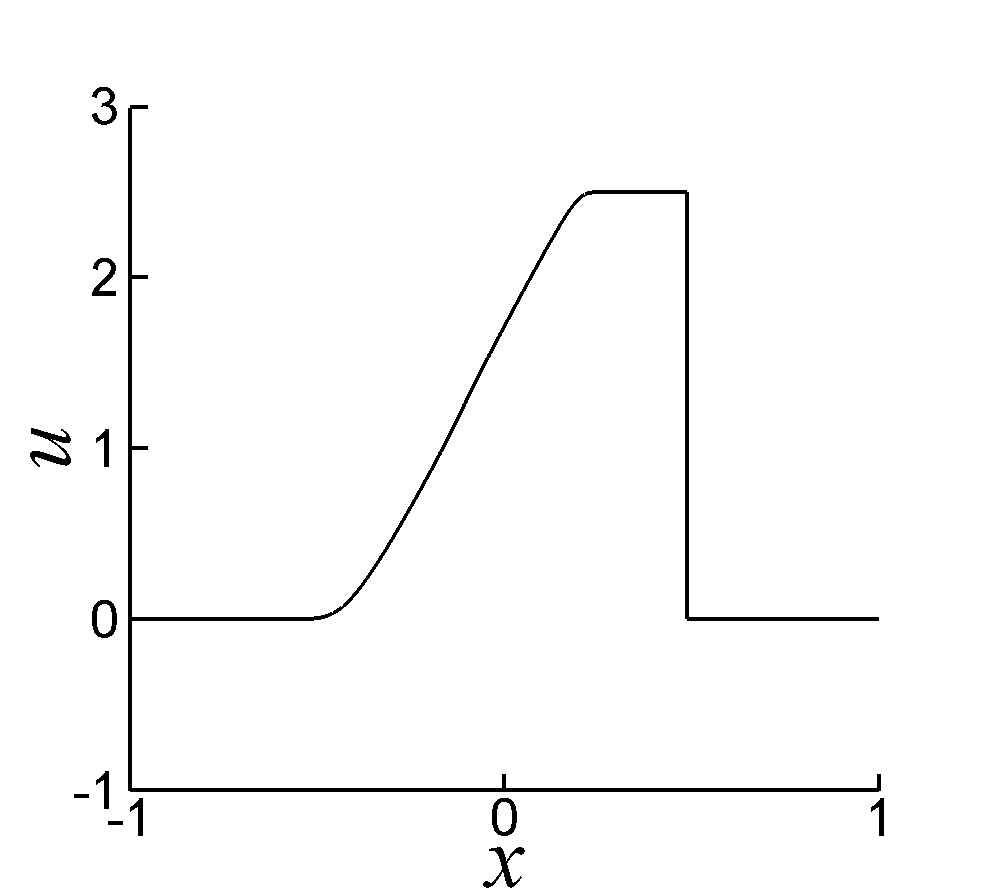}
\end{center} \end{minipage}
\begin{minipage}{0.48\linewidth} \begin{center}
  \includegraphics[width=.9\linewidth]{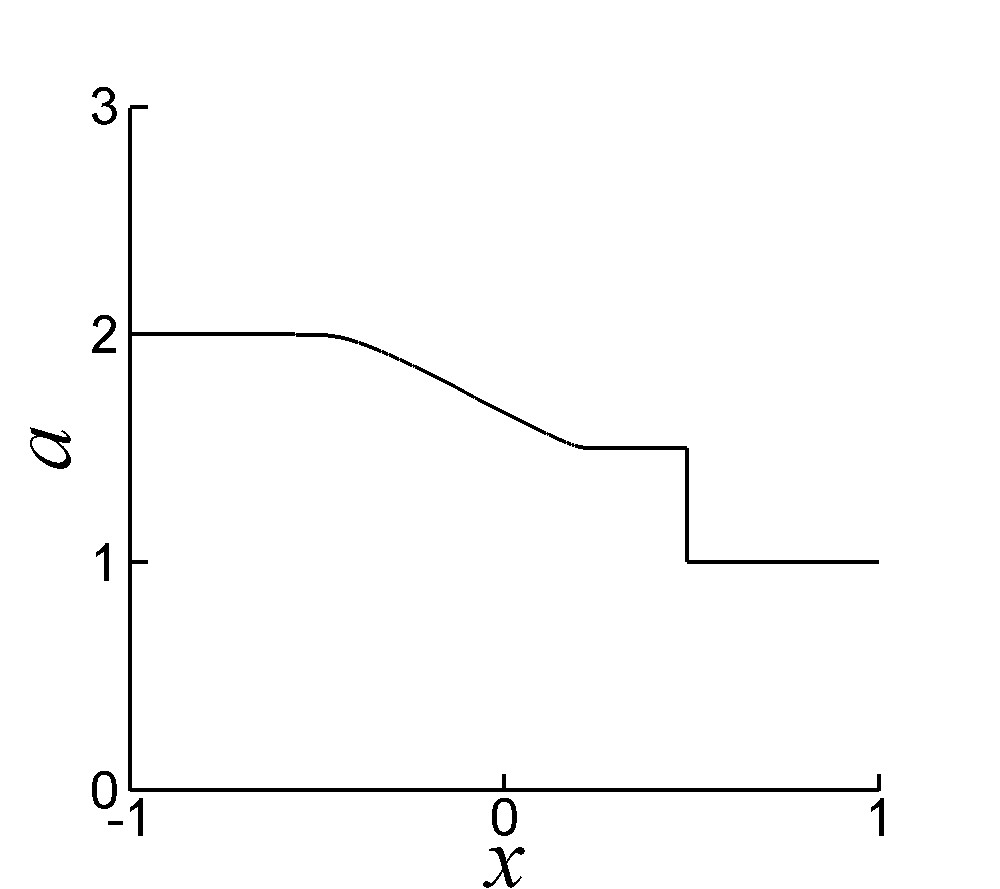}
\end{center} \end{minipage}\\
\begin{minipage}{0.48\linewidth}\begin{center} (a) \end{center} \end{minipage}
\begin{minipage}{0.48\linewidth}\begin{center} (b) \end{center}
\end{minipage}\vspace{-2mm}
\caption{The solution to the perturbed Riemann problem.  The initial
conditions are Equation (\ref{Example1}) that have been filtered
with the Helmholtz filter. There is clearly a shock and an expansion
wave.  This is the solution at time $t=0.2.$ }
\label{riemannproblemfigure4}
\end{center}
\end{figure}

To summarize, the solutions to the Riemann problem take the form of
Equations \eref{reimannsolution} and produces one or two traveling
discontinuities, some of which may be unstable. If instead the
discontinuities in the initial conditions are replace with steep,
but continuous gradients, then the results will begin to include
expansion waves in addition to the shocks.

\section{Numerics} \label{Numericssection}
In section \ref{RiemannNumericsSection} a numerical method is
presented that will solve the Riemann problem.  Here a numerical
method is established that can be applied to more general initial
conditions.  In section \ref{Comparisonsection} we compare the
behavior of the CAHE equations to the homentropic Euler equations.
To accomplish this we use a numerical technique based on the
characteristic structure of the CAHE equations. This technique is
similar in nature to that used by Bhat and Fetecau
\cite{BhatHS:08a}. Since the equations we are investigating are
inspired by the use of averaged characteristics, our numerical
method will track the characteristics as they evolve in time.

Our numerical method is interested in solving Equations \eref{CAHE}
with initial conditions
\begin{eqnarray}
v^+(x,0)=v^+_0(x)\\
v^-(x,0)=v^-_0(x).
\end{eqnarray}
Since we cannot operate on an infinite domain we restrict the
problem to the domain $[a,b]$ and assume that on the boundaries
\begin{eqnarray}
\frac{\p}{\p x} v^+(x,t)=0\\
\frac{\p}{\p x} v^-(x,t)=0.
\end{eqnarray}

From Equations (\ref{positivecharacteristic}) and
(\ref{negativecharacteristic}) we know that the value of $v^\pm$
does not change along its characteristics, thus
\begin{equation}
\label{constantalongcharacteristic}
v^\pm(\phi^\pm(X))=v^\pm_0(X),
\end{equation}
where $X$ is the label of the characteristic.

Let $\xb$ be the equally spaced grid points on the interval
$[-a,a]$.  Let $\textbf{X}^\pm$ be the positions of the
characteristics $X_i$.  These vector can also be considered as the
values of $\phi$ or as a dynamic grid.  At time $t=0$,
$\xb=\textbf{X}^+=\textbf{X}^-$.  Notationally let $v^\pm(\xb)$ be a
vector containing the values of $v^\pm$ evaluated at points $\xb$
and let $v^\pm(\textbf{X}^\pm)$ be $v^\pm$ evaluated at the
locations of the characteristics $\textbf{X}^\pm$.

The positions of the characteristics are iterated in time, thus
giving us the solution.  A typical iteration proceeds as follows.
\begin{itemize}
\item[1.] The values of $v^\pm(\textbf{X}^\pm)$ are known and remain constant through time as a result of Equation \eref{constantalongcharacteristic}. Using cubic splines these values are interpolated to $v^\pm(\xb)$.
\item[2.] From $v^\pm(\xb)$, it is a simple task of finding the values of $u$ and $a$ on the grid $\xb$.  The Helmholtz operator is then numerically inverted to give the values of $\ubar$ and $\abar$ on $\xb$.  These values are then interpolated back to give the values of $\ubar \pm \abar$ on $\textbf{X}^\pm$.
\item[3.] The positions of the characteristics, $\textbf{X}^\pm$, are then iterated in time by the values of $\ubar \pm \abar$ on $\textbf{X}^\pm$ calculated in the previous step.
\end{itemize}

In step 2, the Helmholtz operator is numerically inverted.  To give
a brief explanation of how the operator is inverted first consider a
finite difference approximation of a second derivative on a
equispaced grid.
\begin{equation}
\frac{\p^2}{\p x^2} f(x(i))=\frac{1}{12 \Delta
x}\left[-1f(x_{i-2})+16f(x_{i-1})-30f(x_i)+16f(x_{i+1})-1f(x_{i+2})
\right]
\end{equation}
Thus the Helmholtz operator $(1-\alpha^2 \frac{\p^2}{\p x^2})$ can
be represented as a quintdiagonal matrix.  This is easily inverted
allowing the calculation of $\ubar\pm\abar$ from $u \pm a$ on an
equispaced grid.

A benefit to this method is that plotting the vectors
$\textbf{X}^\pm$ versus time gives a graph of the characteristic
plane.

\section{Comparisons with homentropic Euler equations}\label{Comparisonsection}
The CAHE equations were based on the homentropic Euler equations so
naturally we  compare the behavior of the two.  In particular this
section examines two examples chosen specifically to compare and
contrast the equations behavior. Example 2 was chosen so that the
CAHE equations would have a single traveling shock.  The initial
conditions for Example 2 are
\begin{subequations}\label{Example2}
\begin{eqnarray}
u(x,0)=\left\{\begin{array}{ll}
5 & x<0 \\
0 & x\geq 0
\end{array}\right.\\
a(x,0)=\left\{\begin{array}{ll}
2 & x<0 \\
1 & x\geq 0
\end{array}\right.
\end{eqnarray}
\end{subequations}
Example 3 was chosen so that the homentropic Euler equations would
have a single traveling shock.  The initial conditions for Example 3
are
\begin{subequations} \label{Example3}
\begin{eqnarray}
u(x,0)=\left\{\begin{array}{ll}
0 & x<0 \\
-9.37440483... & x\geq 0
\end{array}\right.\\
a(x,0)=\left\{\begin{array}{ll}
2 & x<0 \\
1 & x\geq 0
\end{array}\right.
\end{eqnarray}
\end{subequations}
 Examples 2a and 3a are the simulation of the
homentropic Euler equations and 2b and 3b address the CAHE
equations. In both examples the two sets of equations are found to
behave significantly differently. This is attributed to the the fact
that the CAHE equations will not produce new characteristics as
shown by the proof of Theorem \ref{existencetheorem}, while the
homentropic Euler equations do produce new characteristics.

\subsection{Numerics for the homentropic Euler equations}
Numerical simulations of the homentropic Euler equations require a
separate method than that described in section \ref{Numericssection}
because characteristics in the homentropic Euler equations collide
and are created. For the numerical simulations of the homentropic
Euler equations the Richtmyer method, a well-established if low
order method, was utilized as described by \cite{LaneyCB:98a}. This
method is second order finite difference scheme and employs an
artificial viscosity. Clearly this is not the most optimal numerical
scheme for the Euler equations, but run at a sufficiently high
resolution it will suit our purposes. For reference there were
$2^{12}$ grid points on a $[-1, 1]$ domain.  This method requires an
artificial viscosity for stability when examining the Riemann
problem.  Several different values of $\nu$ were tested to see that
the value did not significantly affect the solutions on the time
interval examined. For the numerical simulations shown here, the
artificial viscosity was set at $\nu=0.08$.

\subsection{Example 2}
Example 2 \eref{Example2} was chosen so that the CAHE equations
would have a single traveling shock. Notice that with these initial
conditions there is a discontinuity in the variable $v^+$, but $v^-$
is constant.  Since $v^-$ is constant in the beginning it should
remain constant for all time. This is a consequence of the fact that
no characteristics are created or destroyed as time progresses.

\subsubsection{Example 2b, the CAHE equations} Figure \ref{example2b} shows the
simulation for Example 2b which examine the CAHE equations. The
simulation was conducted with a resolution of $2^{12}$ and
$\alpha=0.02$.  Figures \ref{example2b}a and \ref{example2b}b show
the single discontinuity progressing to the right as expected.
Figures \ref{example2b}c and \ref{example2b}d show the values of
$v^+$ and $v^-$.   The discontinuity in $v^+$ has traveled to the
right and as expected $v^-$ remains constant.

Figures \ref{example2b}e and \ref{example2b}f show the paths of the
$v^+$ and $v^-$ characteristics respectively.  Figure
\ref{example2b}e shows that the characteristics of $v^+$ are moving
towards each other and as they near the shock are bent toward each
other.  They do not, however, intersect. Figure \ref{example2b}f
shows that the characteristics of $v^-$ pass through the shock and
upon doing so change their speed to match the characteristics on the
other side.

\begin{figure}[!ht]
\begin{center}
\begin{minipage}{0.38\linewidth} \begin{center}
  \includegraphics[width=.9\linewidth]{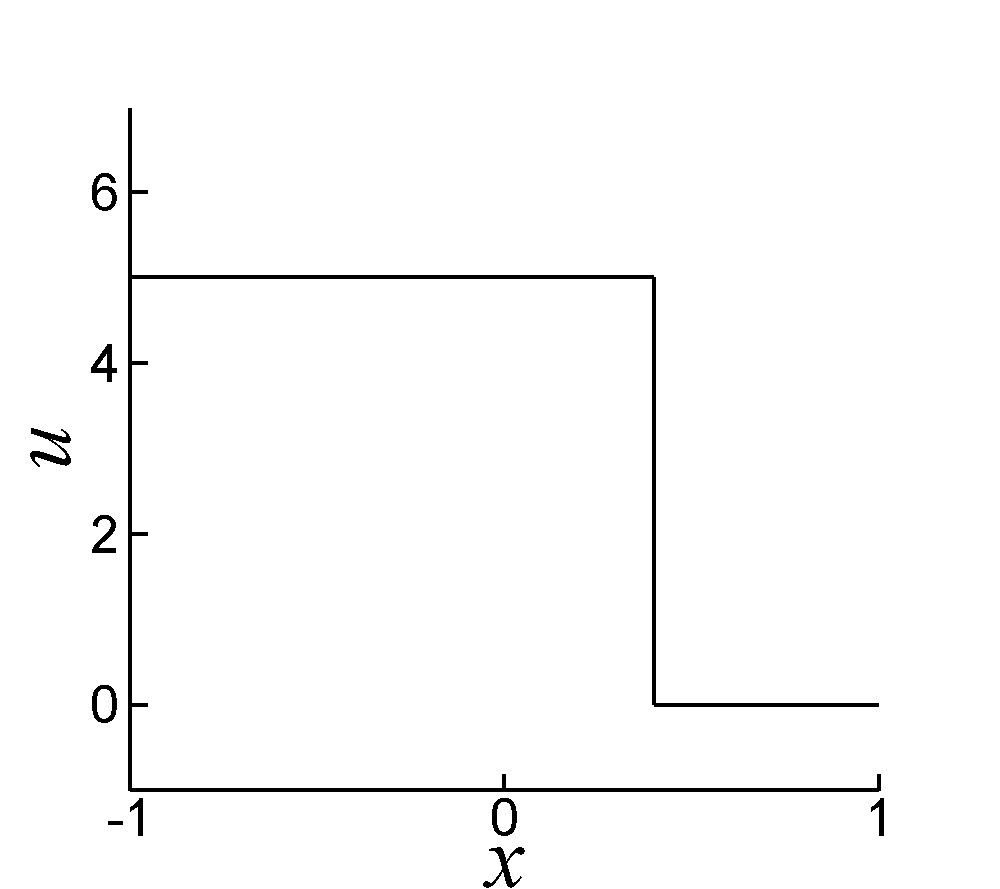}
\end{center} \end{minipage}
\begin{minipage}{0.38\linewidth} \begin{center}
  \includegraphics[width=.9\linewidth]{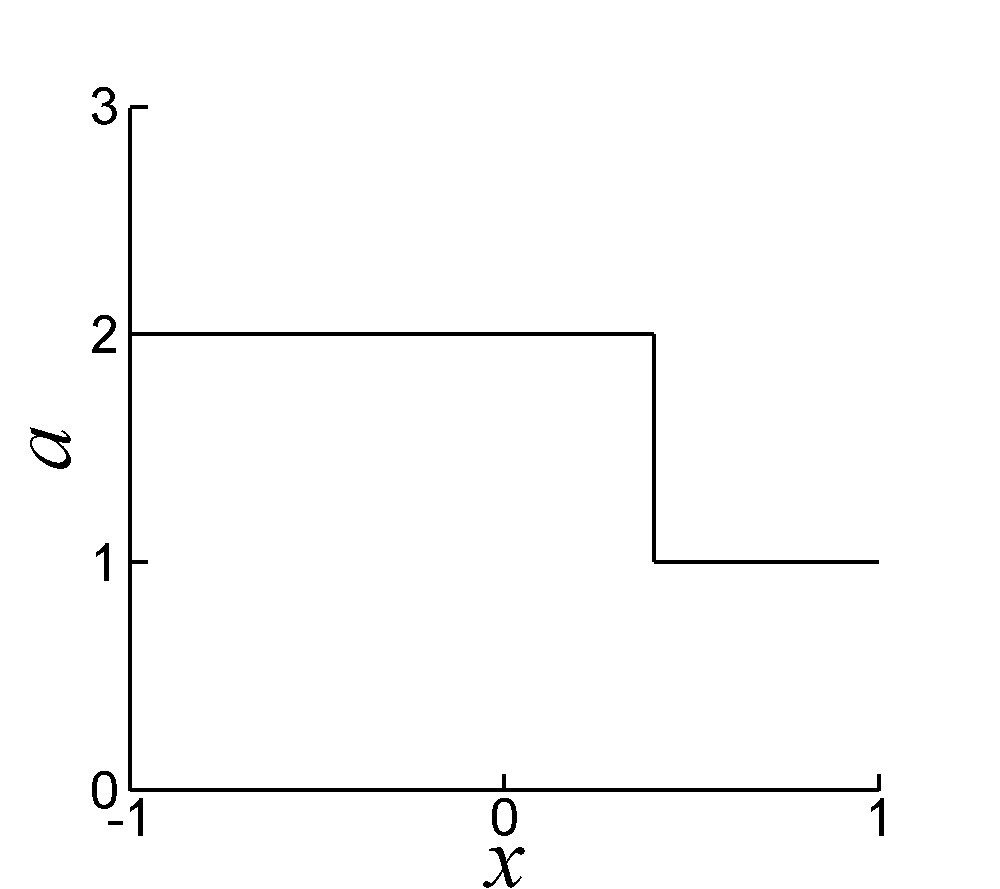}
\end{center} \end{minipage}\\ \vspace{2mm}
\begin{minipage}{0.38\linewidth}\begin{center} (a) \end{center} \end{minipage}
\begin{minipage}{0.38\linewidth}\begin{center} (b) \end{center}
\end{minipage}\\
\begin{minipage}{0.38\linewidth} \begin{center}
  \includegraphics[width=.9\linewidth]{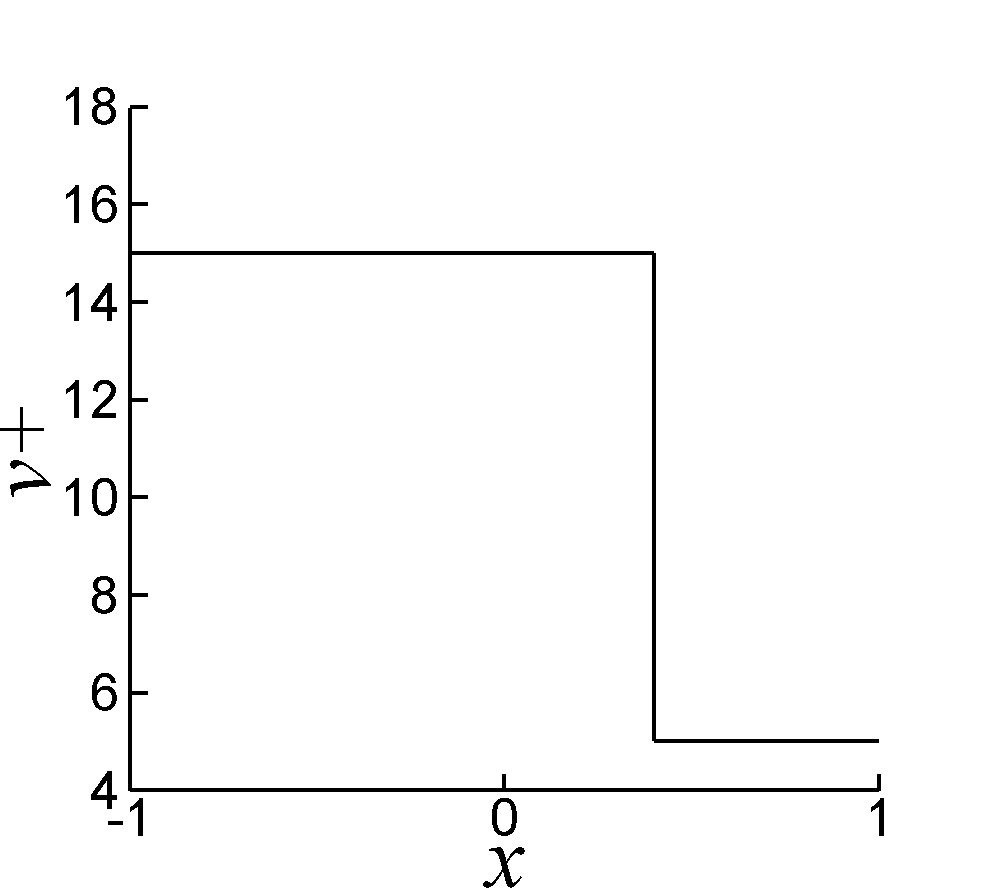}
\end{center} \end{minipage}
\begin{minipage}{0.38\linewidth} \begin{center}
  \includegraphics[width=.9\linewidth]{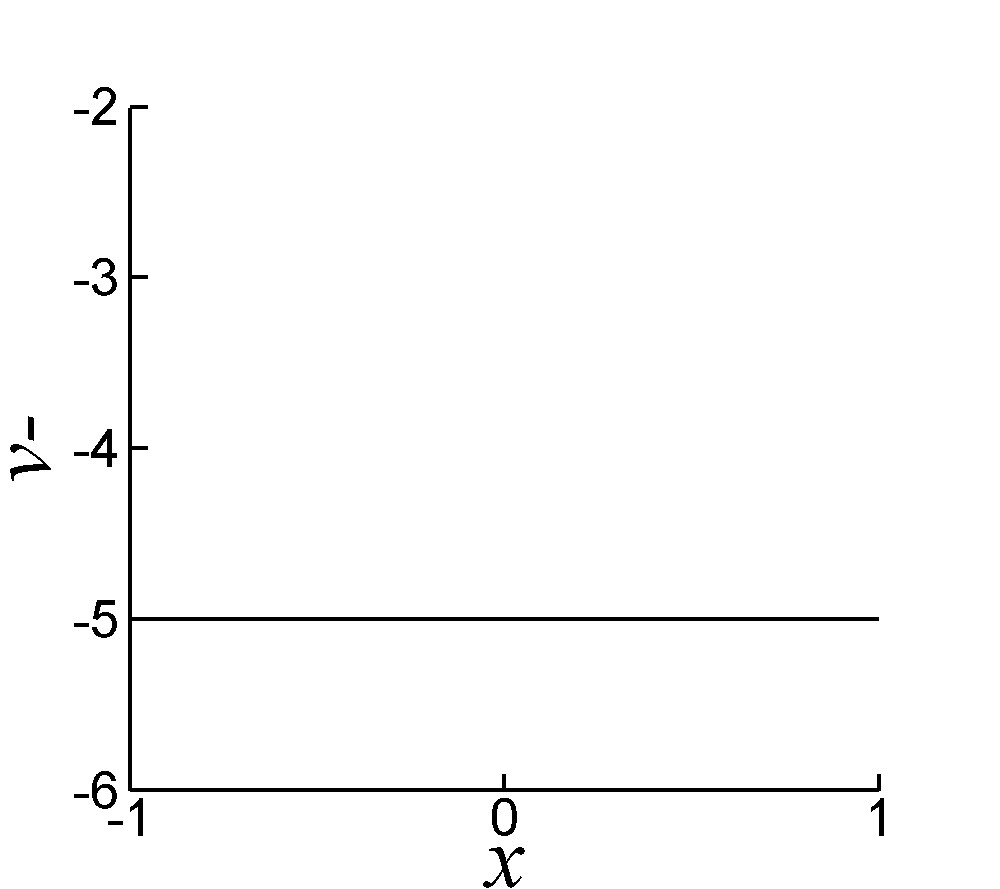}
\end{center} \end{minipage}\\ \vspace{2mm}
\begin{minipage}{0.38\linewidth}\begin{center} (c) \end{center} \end{minipage}
\begin{minipage}{0.38\linewidth}\begin{center} (d) \end{center}
\end{minipage}\\
\begin{minipage}{0.38\linewidth} \begin{center}
  \includegraphics[width=.9\linewidth]{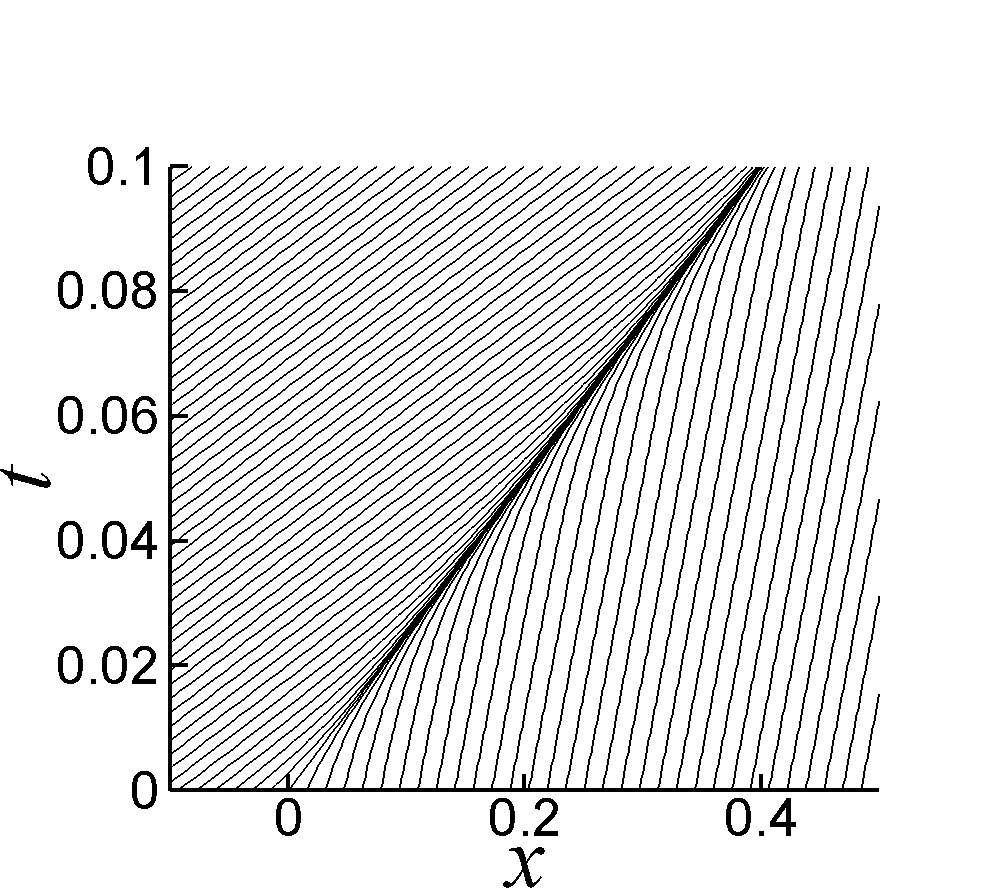}
\end{center} \end{minipage}
\begin{minipage}{0.38\linewidth} \begin{center}
  \includegraphics[width=.9\linewidth]{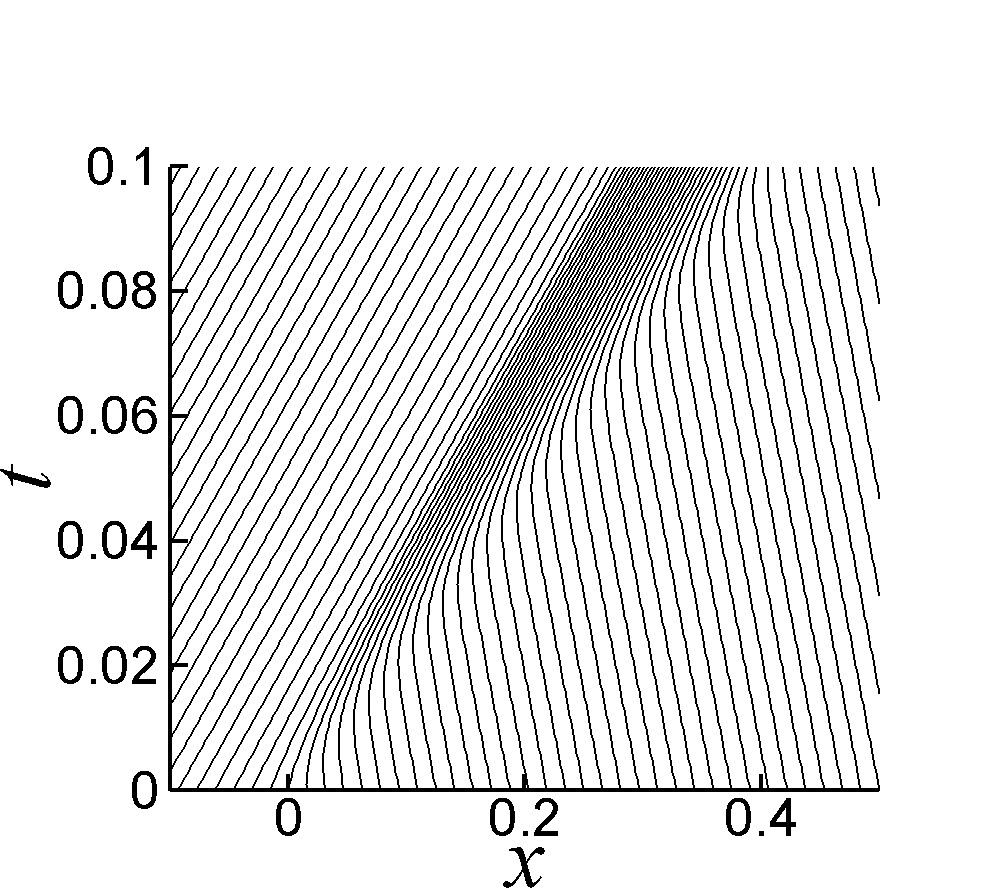}
\end{center} \end{minipage}\\ \vspace{2mm}
\begin{minipage}{0.38\linewidth}\begin{center} (e) \end{center} \end{minipage}
\begin{minipage}{0.38\linewidth}\begin{center} (f) \end{center}
\end{minipage}\vspace{-2mm}
\caption{Example 2b at time $t=0.1$.  This example was chosen so
that the CAHE equations would have a single traveling shock which is
clearly visible in the velocity and speed of sound in (a) and (b).
(c) and (d) show graphs of the Riemann invariants $v^+$ and $v^-$.
Note that $v^-$ is constant as it was in the initial conditions. (e)
and (f) show the $v^+$ and $v^-$ characteristics respectively.
Notice that the $v^+$ characteristics are converging to the shock
but never intersect, while the $v^-$ characteristics pass through
the shock and change speeds as they do so. } \label{example2b}
\end{center}
\end{figure}

\subsubsection{Example 2a, the homentropic Euler equations}
Now we examine the behavior of homentropic Euler equations for the
same initial conditions.  Figure \ref{example2a} shows the
simulation for Example 2a.  The simulation was conducted with a
resolution of $2^{12}$.

Figures \ref{example2a}a and \ref{example2a}b show a shock
progressing to the right as expected. In addition there is an
expansion wave also occurring. Figures \ref{example2a}c and
\ref{example2a}d show the values of $v^+$ and $v^-$.  Of primary
notice is that $v^-$ is no longer a constant but has attained new
value.  This is a significant departure from the behavior of the
CAHE equations. Figures \ref{example2a}e and \ref{example2a}f are
not graphs of the actual simulation.  The Richtmyer method does not
lend itself to characteristic graphs as does the method used for the
CAHE equations. Instead they are sketches that depict the behavior
of the giving simulation.

The characteristics in $v^+$ change speed as they travel through the
expansion wave.  The characteristics intersect causing the shock.
Thus at the shock the characteristics are destroyed.

For the $v^-$ characteristics, at time $t=0$, the characteristics to
the left of the discontinuity have a speed greater than those to the
right, but both are less than the speed of the shock. The result is
that the characteristics on the right will be absorb into the shock
and between the characteristics on the left and the shock there will
be a gap devoid of characteristics. This gap is filled with an
expansion wave created at $t=0$ and new characteristics that are
created at the location of the shock continually as time progresses.
The values of these new characteristics can be determined with the
Rankine-Hugoniot jump conditions.  This creation of characteristics
again is a behavior that the CAHE equations does not demonstrate.

\begin{figure}[!ht]
\begin{center}
\begin{minipage}{0.38\linewidth} \begin{center}
  \includegraphics[width=.9\linewidth]{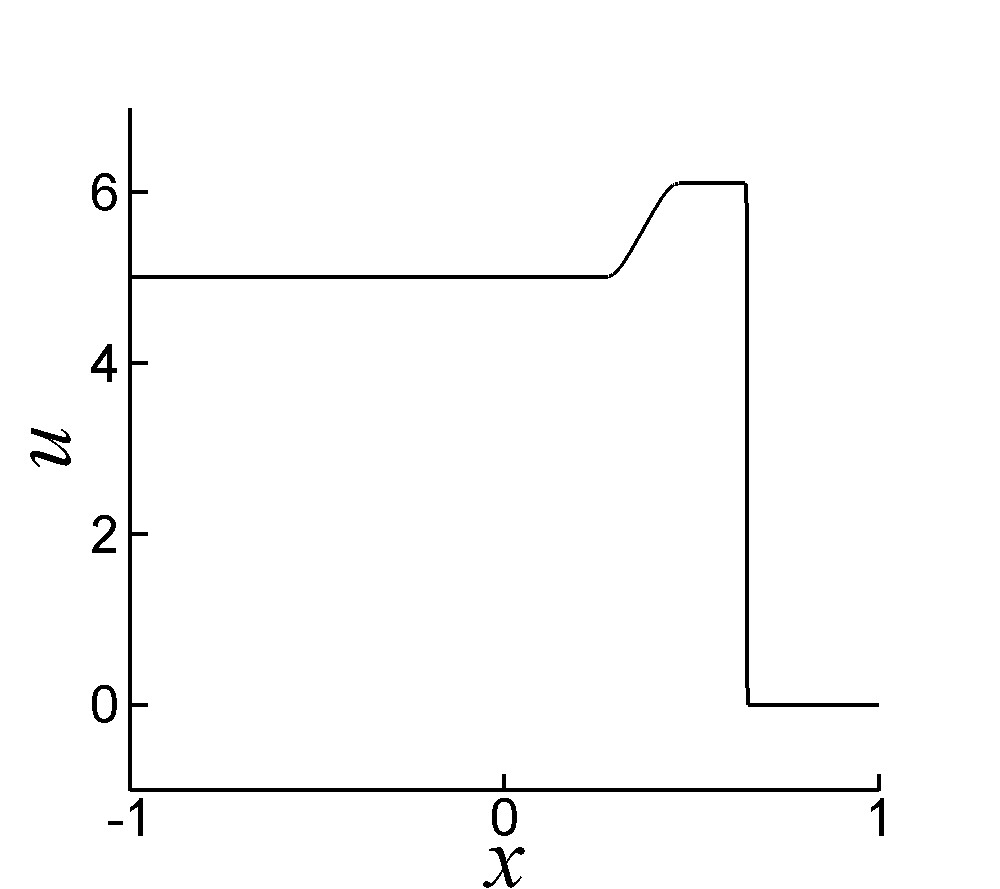}
\end{center} \end{minipage}
\begin{minipage}{0.38\linewidth} \begin{center}
  \includegraphics[width=.9\linewidth]{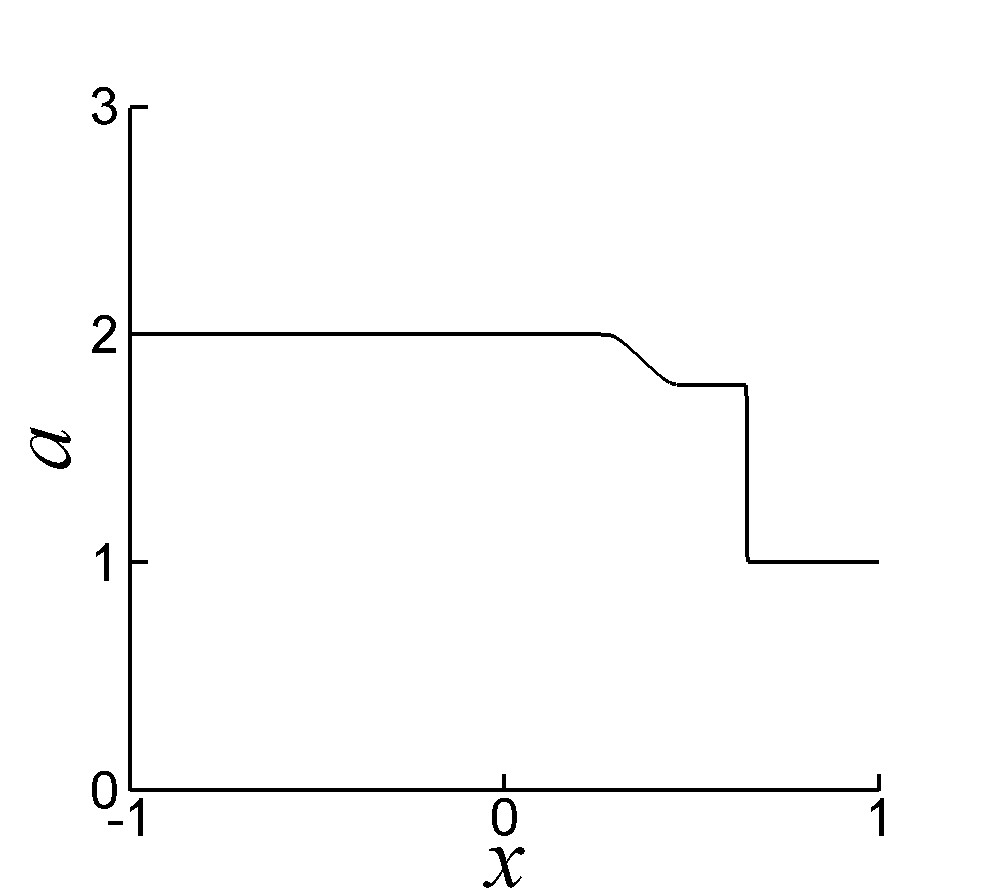}
\end{center} \end{minipage}\\ \vspace{2mm}
\begin{minipage}{0.38\linewidth}\begin{center} (a) \end{center} \end{minipage}
\begin{minipage}{0.38\linewidth}\begin{center} (b) \end{center}
\end{minipage}\\
\begin{minipage}{0.38\linewidth} \begin{center}
  \includegraphics[width=.9\linewidth]{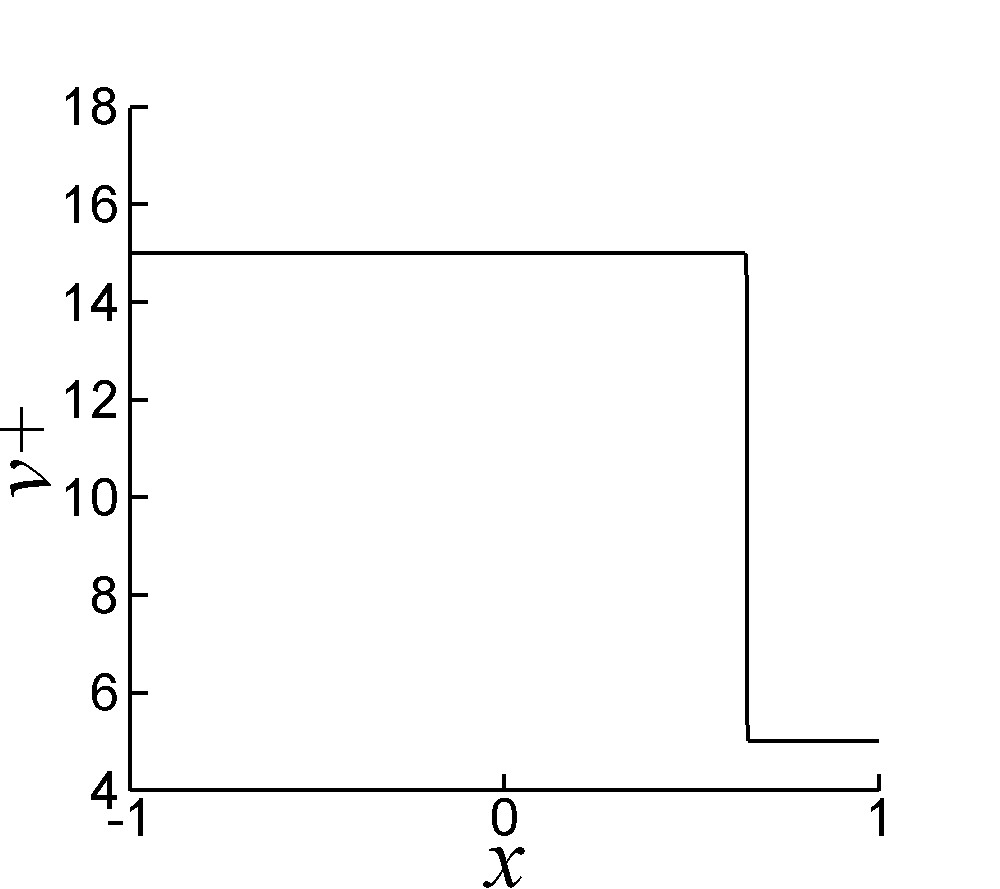}
\end{center} \end{minipage}
\begin{minipage}{0.38\linewidth} \begin{center}
  \includegraphics[width=.9\linewidth]{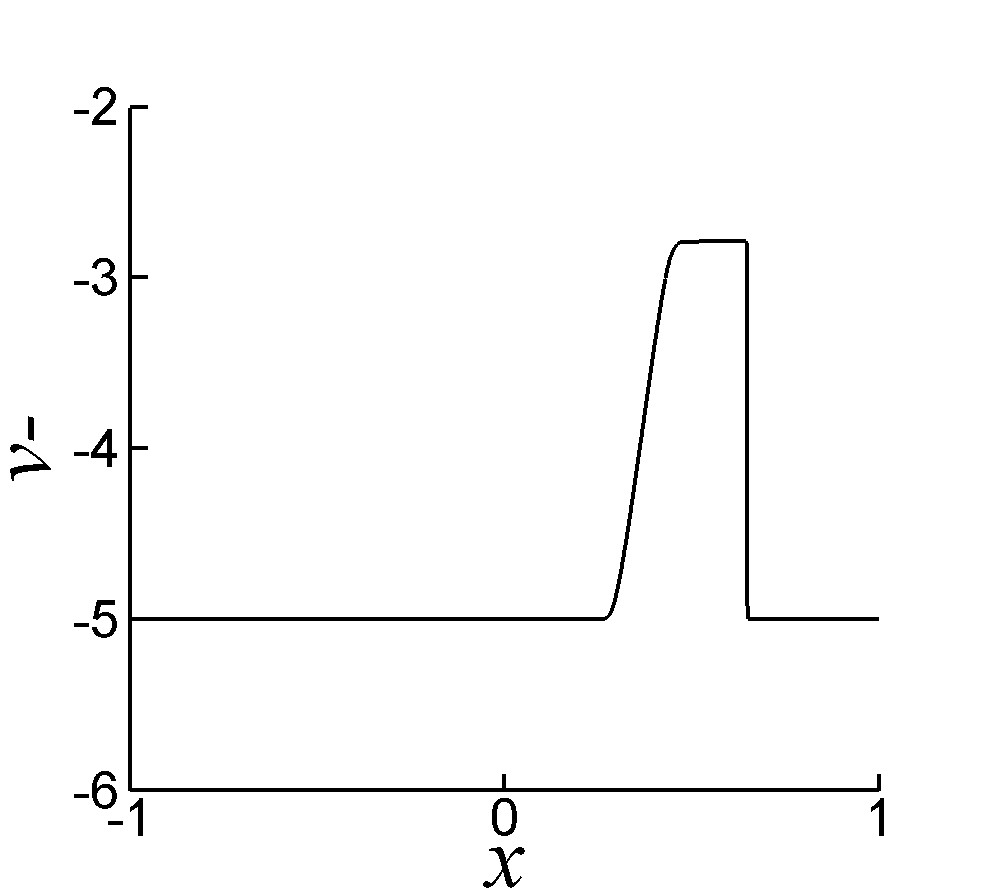}
\end{center} \end{minipage}\\ \vspace{2mm}
\begin{minipage}{0.38\linewidth}\begin{center} (c) \end{center} \end{minipage}
\begin{minipage}{0.38\linewidth}\begin{center} (d) \end{center}
\end{minipage}\\
\begin{minipage}{0.38\linewidth} \begin{center}
  \includegraphics[width=.9\linewidth]{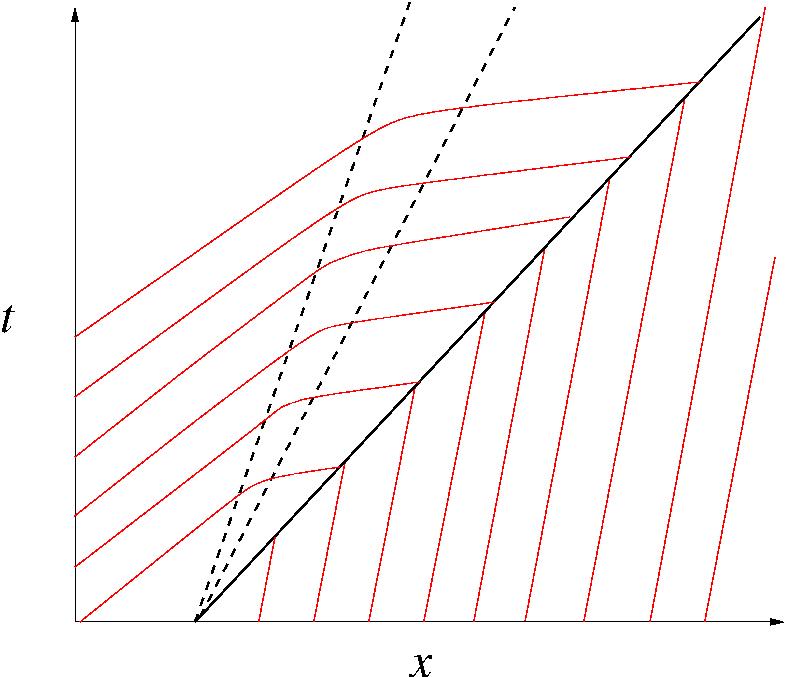}
\end{center} \end{minipage}
\begin{minipage}{0.38\linewidth} \begin{center}
  \includegraphics[width=.9\linewidth]{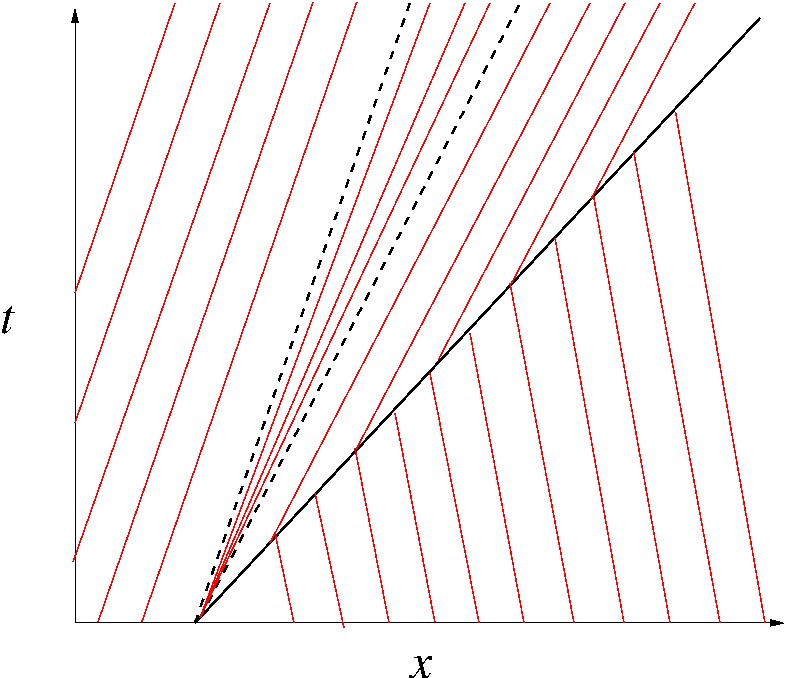}
\end{center} \end{minipage}\\ \vspace{2mm}
\begin{minipage}{0.38\linewidth}\begin{center} (e) \end{center} \end{minipage}
\begin{minipage}{0.38\linewidth}\begin{center} (f) \end{center}
\end{minipage}\vspace{-2mm}
\caption{Example 2a at time $t=0.1$.  While the CAHE equations form
a single shock the homentropic Euler equations clearly form a shock
and expansion wave.  (a) and (b) show the shocks and expansion wave
clearly in the velocity and speed of sound.  (c) and (d) show the
Riemann invariants $v^+$ and $v^-$.  Notice that while $v^-$ began
as a constant it is no longer.  This is due to the creation of new
characteristics at the shock.  (e) and (f) show sketches of the
characteristics for this example.  (e) shows that the $v^+$
characteristics are being absorbed by the shock.  (f) shows the
$v^-$ characteristics being created at the shock. The dotted line
represents the expansion wave. } \label{example2a}
\end{center}
\end{figure}

\subsection{Example 3}
The previous example was chosen so that the CAHE equations would
have a single shock.  In contrast, Example 3 \eref{Example3} was
chosen so that the homentropic Euler equations would have a single
traveling shock. Again a significant departure in behavior will be
noticed.

\subsubsection{Example 3a, the homentropic Euler equations}
First we examine the behavior of homentropic Euler equations for
Example 3a. Figure \ref{example3a} shows the simulation for Example
3a. The simulation was conducted with a resolution of $2^{12}$.

Figures \ref{example3a}a and \ref{example3a}b show a single shock.
The shock is close to stationary but is progressing to the right.
Figures \ref{example3a}c and \ref{example3a}d show the values of
$v^+$ and $v^-$.   The discontinuity can be seen in both Riemann
invariants.  This is noticeably different than the behavior of the
CAHE equations where a single traveling shock will appear in only
one of the invariants.

As in the previous examples Figures \ref{example3a}e and
\ref{example3a}f are not graphs of the actual simulation, but
sketches that depict the behavior of the giving simulation. The
$v^+$ characteristics to the left of the discontinuity travel faster
than those to the right with the speed of the shock found in between
those speeds. Thus the characteristics collide and cause the shock.
The $v^-$ characteristics to the left of the discontinuity travel
faster than those to the right, however both are slower than the
speed of the shock.  Thus there is an area devoid of characteristics
between the left characteristics and the shock.  This is filled with
new characteristics that originate from the shock.

\begin{figure}[!ht]
\begin{center}
\begin{minipage}{0.38\linewidth} \begin{center}
  \includegraphics[width=.9\linewidth]{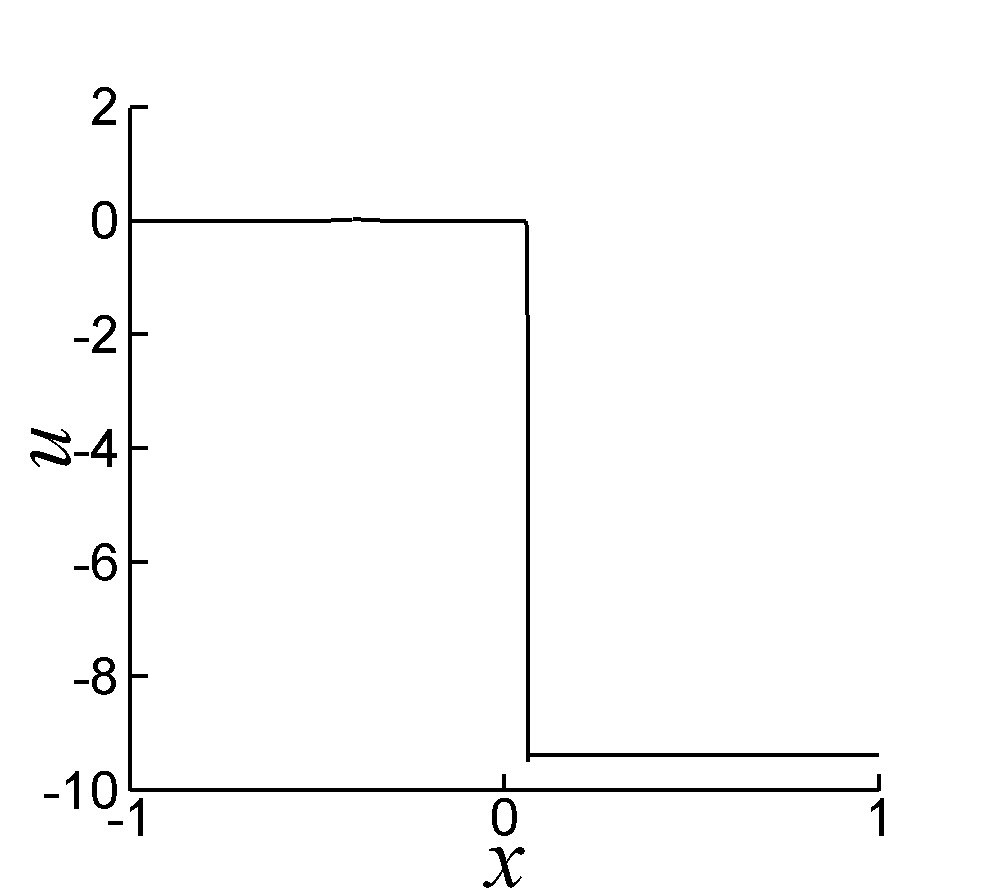}
\end{center} \end{minipage}
\begin{minipage}{0.38\linewidth} \begin{center}
  \includegraphics[width=.9\linewidth]{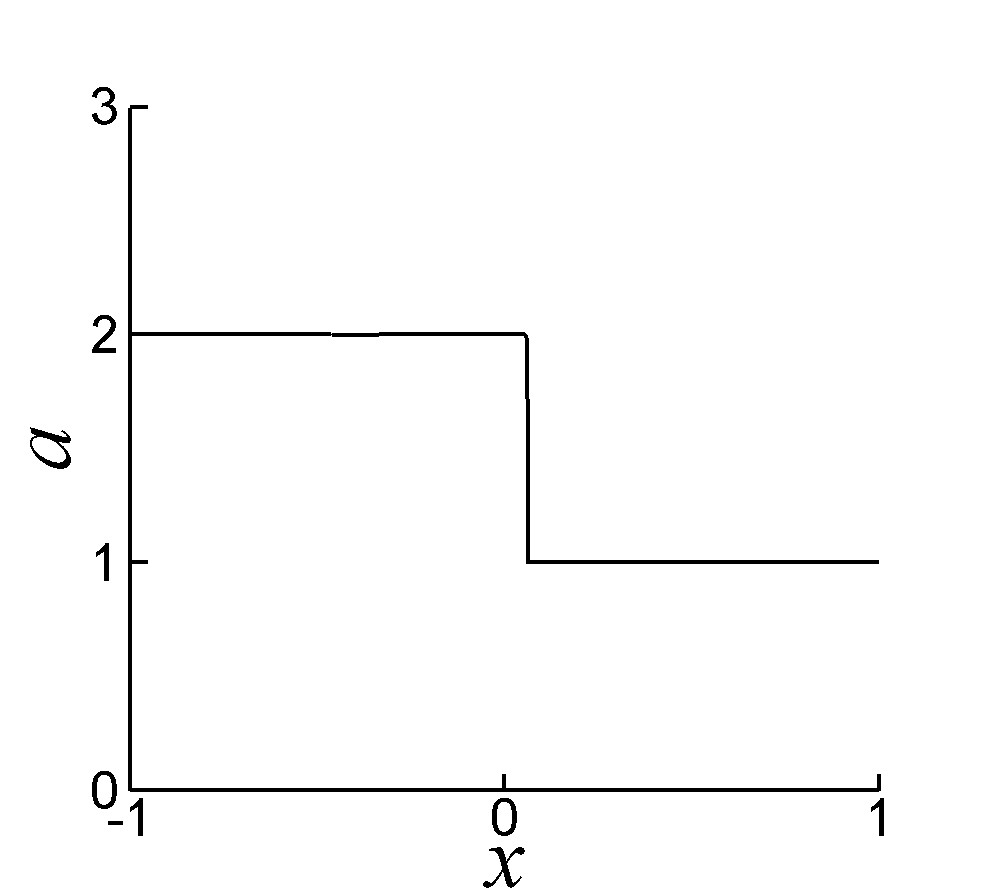}
\end{center} \end{minipage}\\ \vspace{2mm}
\begin{minipage}{0.38\linewidth}\begin{center} (a) \end{center} \end{minipage}
\begin{minipage}{0.38\linewidth}\begin{center} (b) \end{center}
\end{minipage}\\
\begin{minipage}{0.38\linewidth} \begin{center}
  \includegraphics[width=.9\linewidth]{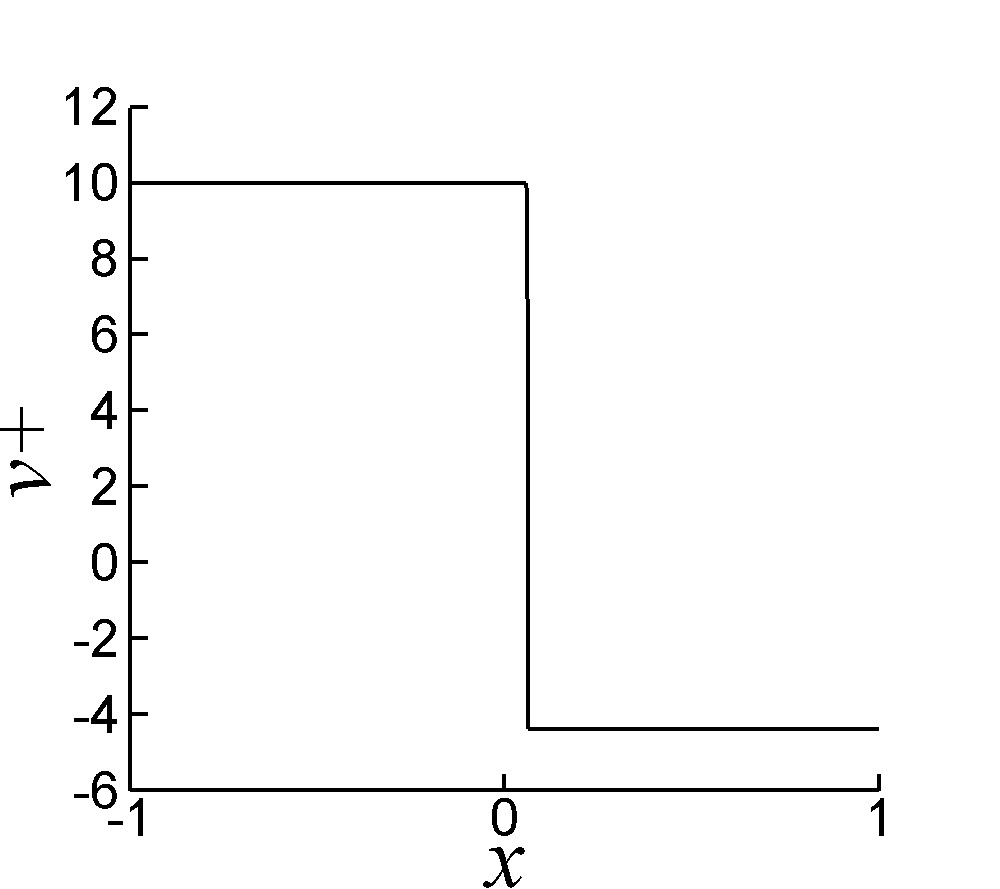}
\end{center} \end{minipage}
\begin{minipage}{0.38\linewidth} \begin{center}
  \includegraphics[width=.9\linewidth]{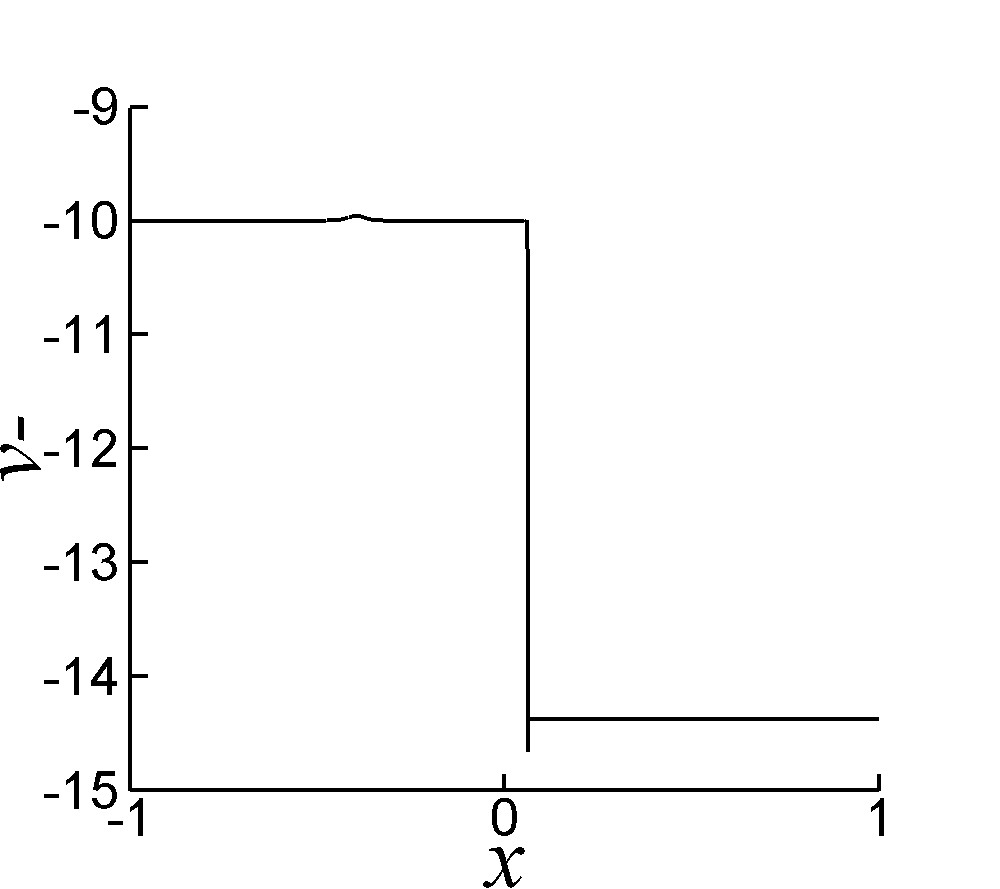}
\end{center} \end{minipage}\\ \vspace{2mm}
\begin{minipage}{0.38\linewidth}\begin{center} (c) \end{center} \end{minipage}
\begin{minipage}{0.38\linewidth}\begin{center} (d) \end{center}
\end{minipage}\\
\begin{minipage}{0.38\linewidth} \begin{center}
  \includegraphics[width=.9\linewidth]{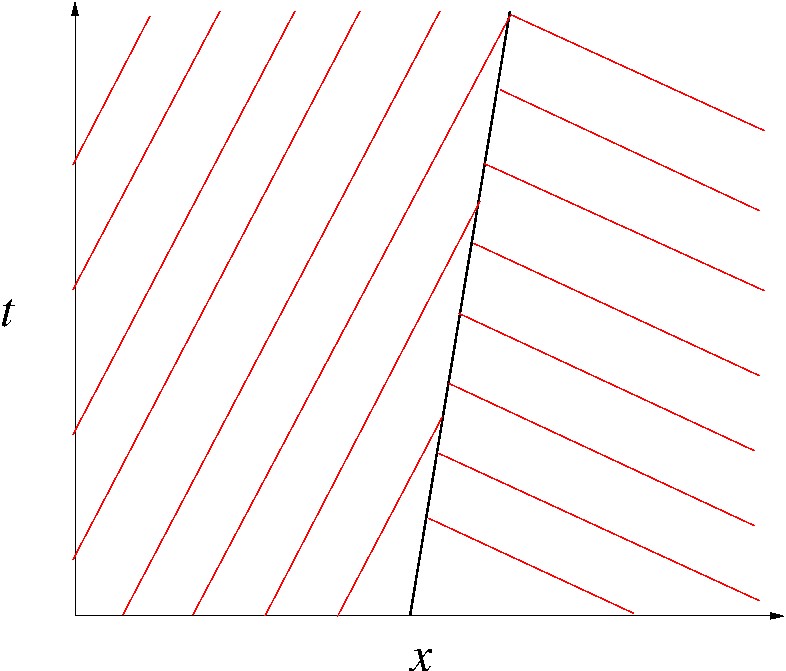}
\end{center} \end{minipage}
\begin{minipage}{0.38\linewidth} \begin{center}
  \includegraphics[width=.9\linewidth]{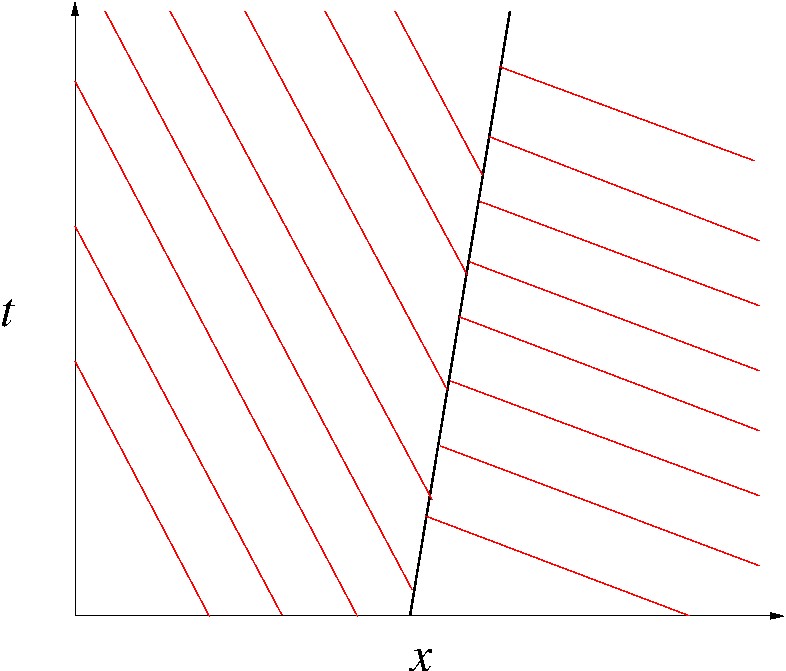}
\end{center} \end{minipage}\\ \vspace{2mm}
\begin{minipage}{0.38\linewidth}\begin{center} (e) \end{center} \end{minipage}
\begin{minipage}{0.38\linewidth}\begin{center} (f) \end{center}
\end{minipage}\vspace{-2mm}
\caption{Example 3a at time $t=0.2$.  This example was chosen to
have a single shock for the homentropic Euler equations which is
clearly visible in (a) and (b).  The shock is moving to the right
but clearly not as quickly as previous examples.  (c) and (d) show
the Riemann invariants $v^+$ and $v^-$.  (e) and (f) are sketches of
the characteristics for $v^+$ and $v^-$.  Again in (f) $v^-$
characteristics are being created at the shock. } \label{example3a}
\end{center}
\end{figure}

\subsubsection{Example 3b, the CAHE equations}
Now we examine the CAHE equations for Example 3b. Notice that with
these initial conditions there is a discontinuity in the both
variables $v^+$ and $v^-$. Thus it is to be expected that there will
be two phenomenon, either expansion waves or shocks.  In this case,
both will be shocks.

Figure \ref{example2b} shows the simulation for Example 3b.  The
simulation was conducted with a resolution of $2^{12}$  and
$\alpha=0.02$.

 Figures \ref{example2b}a and \ref{example2b}b clearly
show two distinct shocks, close together, progressing to the left.
These shocks can also be seen in Figures \ref{example2b}c and
\ref{example2b}d in the variables $v^+$ and $v^-$.

Figures \ref{example2b}e and \ref{example2b}f show the paths of the
$v^+$ and $v^-$ characteristics respectively.  The line across the
center of the graph represents the remapping of characteristics as
they had grown too close together for convenient computation. For
both sets of characteristics there is a convergence of the
characteristics towards the shock. This demonstrates that it is a
stable discontinuity and that a perturbation will not turn these
into expansion waves. Again the characteristics are bent towards
each other, but never intersect.

\begin{figure}[!ht]
\begin{center}
\begin{minipage}{0.38\linewidth} \begin{center}
  \includegraphics[width=.9\linewidth]{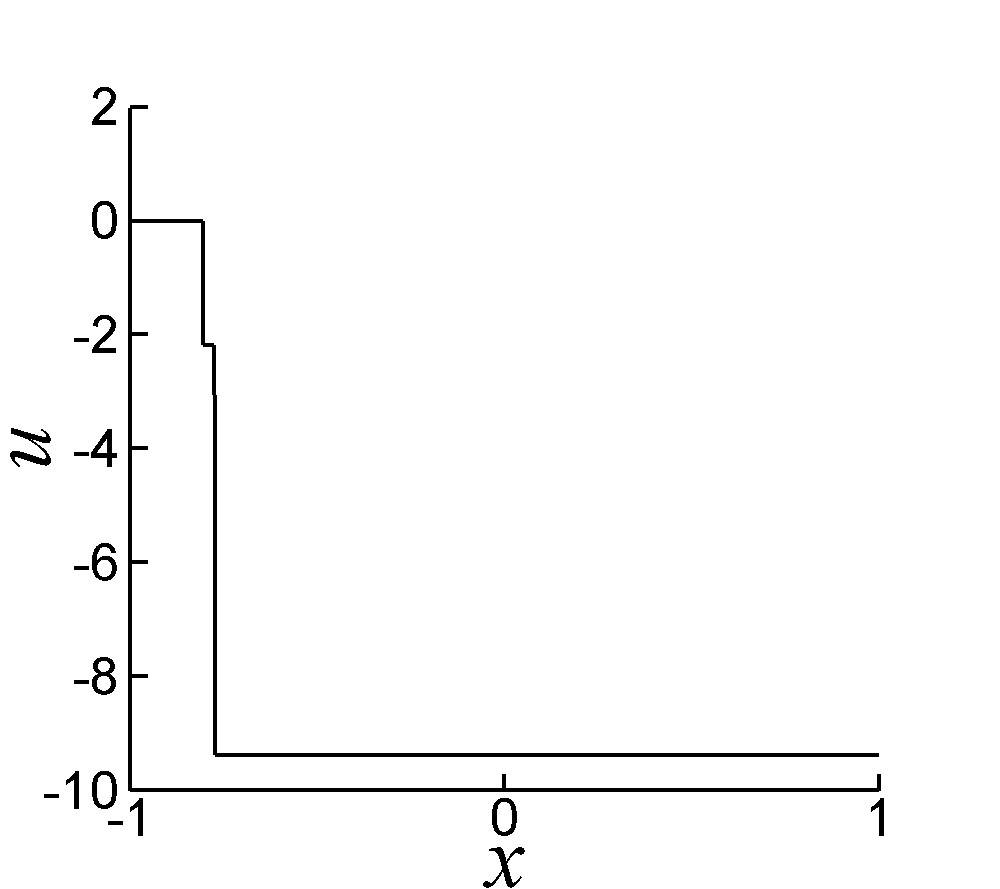}
\end{center} \end{minipage}
\begin{minipage}{0.38\linewidth} \begin{center}
  \includegraphics[width=.9\linewidth]{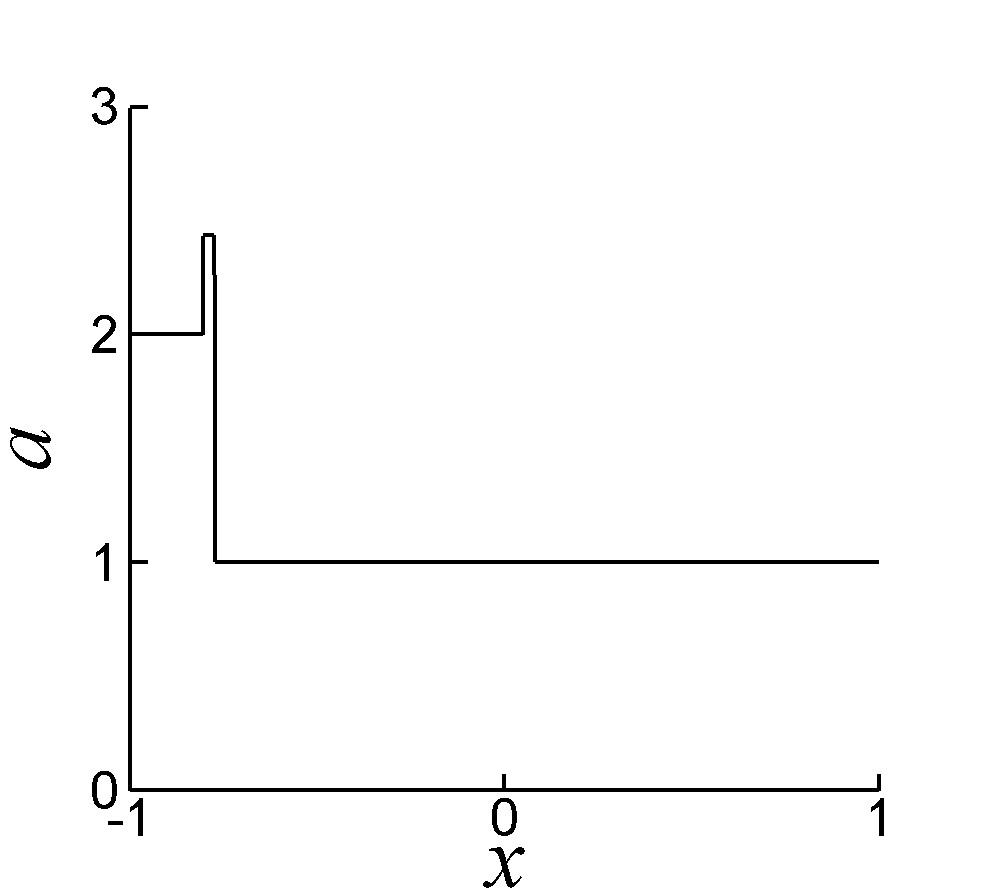}
\end{center} \end{minipage}\\  \vspace{2mm}
\begin{minipage}{0.38\linewidth}\begin{center} (a) \end{center} \end{minipage}
\begin{minipage}{0.38\linewidth}\begin{center} (b) \end{center}
\end{minipage}\\
\begin{minipage}{0.38\linewidth} \begin{center}
  \includegraphics[width=.9\linewidth]{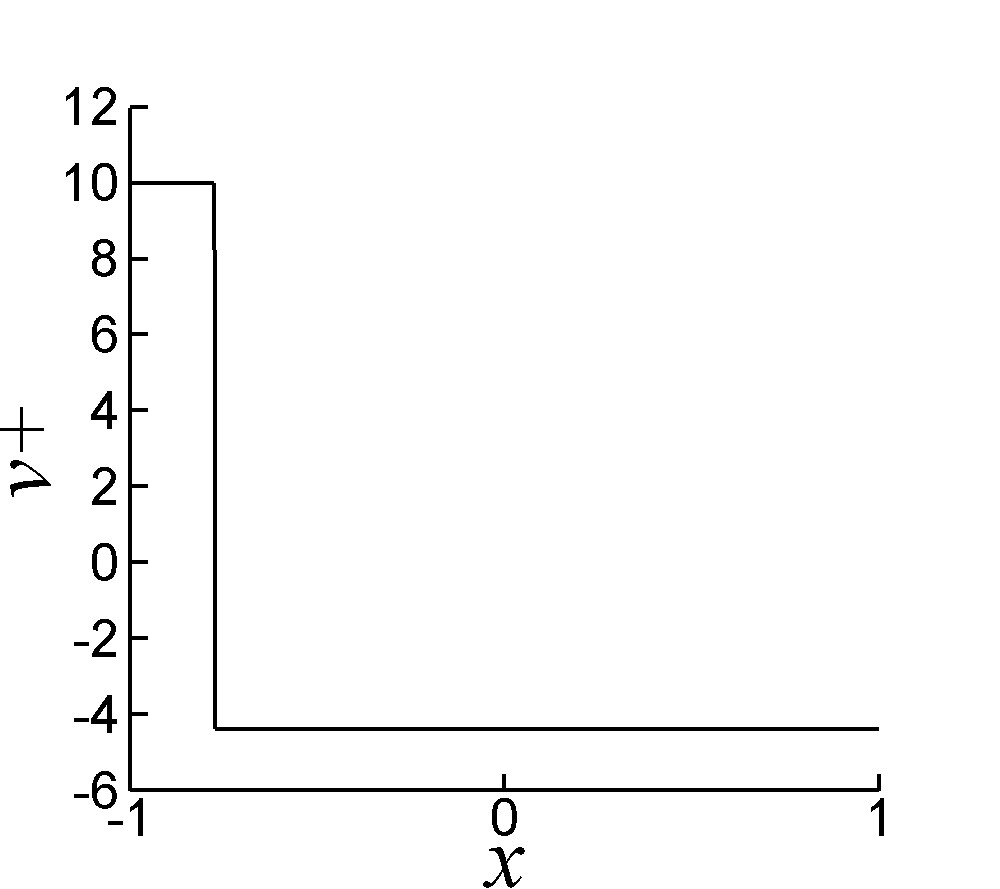}
\end{center} \end{minipage}
\begin{minipage}{0.38\linewidth} \begin{center}
  \includegraphics[width=.9\linewidth]{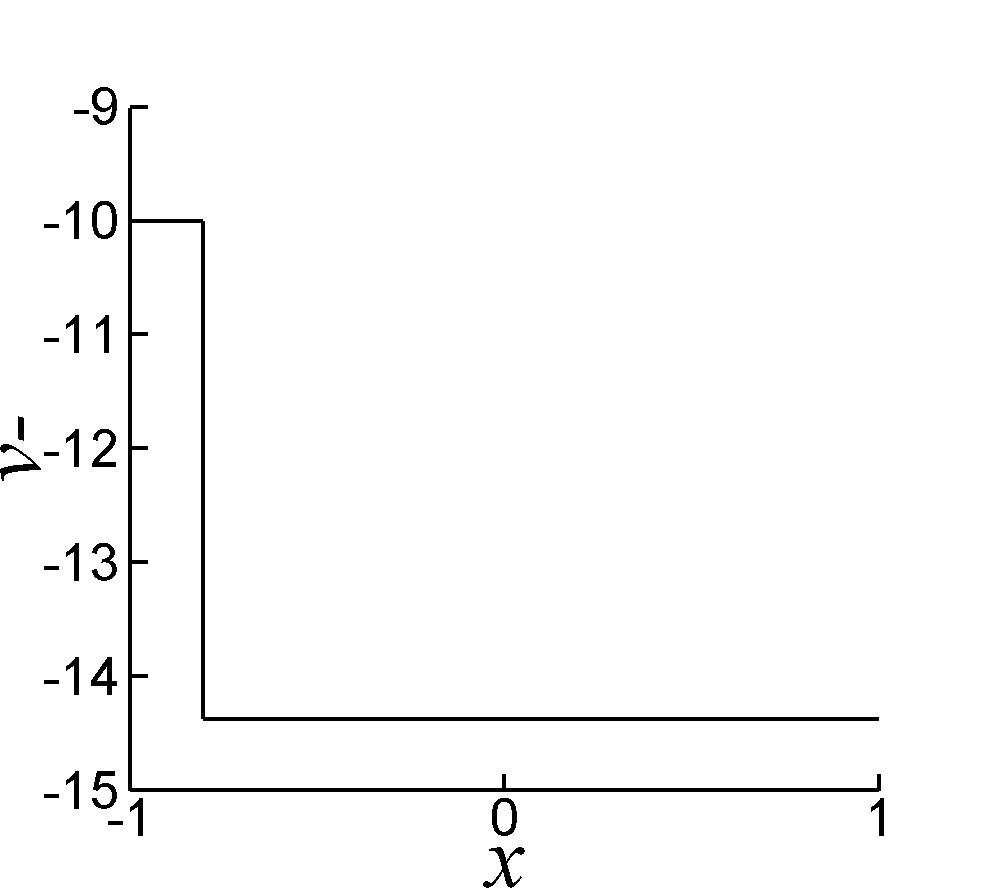}
\end{center} \end{minipage}\\ \vspace{2mm}
\begin{minipage}{0.38\linewidth}\begin{center} (c) \end{center} \end{minipage}
\begin{minipage}{0.38\linewidth}\begin{center} (d) \end{center}
\end{minipage}\\
\begin{minipage}{0.38\linewidth} \begin{center}
  \includegraphics[width=.9\linewidth]{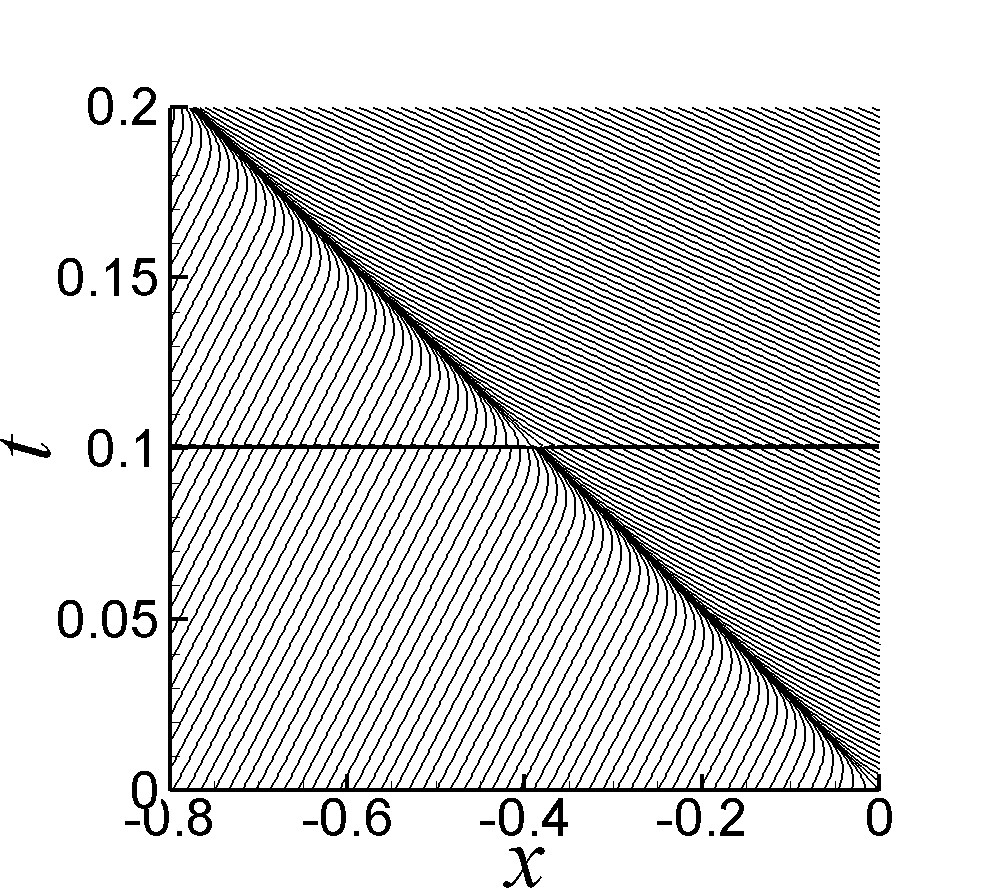}
\end{center} \end{minipage}
\begin{minipage}{0.38\linewidth} \begin{center}
  \includegraphics[width=.9\linewidth]{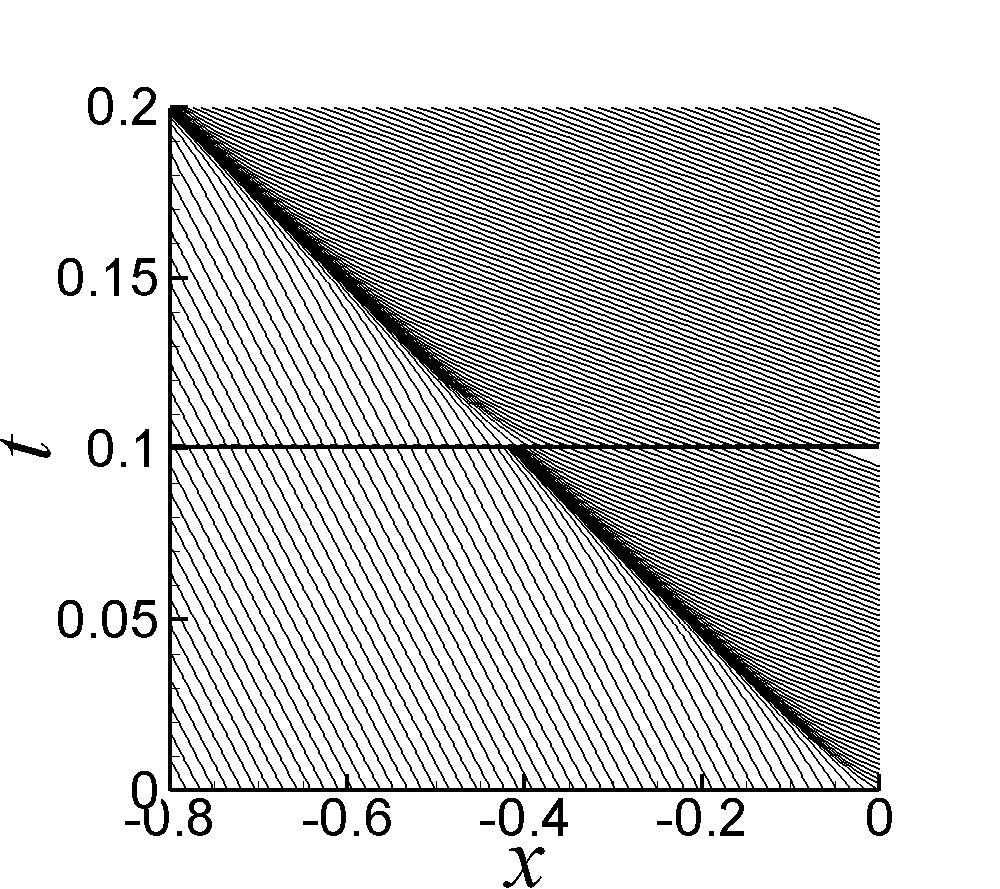}
\end{center} \end{minipage}\\ \vspace{2mm}
\begin{minipage}{0.38\linewidth}\begin{center} (e) \end{center} \end{minipage}
\begin{minipage}{0.38\linewidth}\begin{center} (f) \end{center}
\end{minipage}\vspace{-2mm}
\caption{Example 3b at time $t=0.2$ (a) and (b) show that there are
two distinct left traveling shocks in the velocity and speed of
sound. (c) and (d) show that these shocks exist separately in the
Riemann invariants $v^+$ and $v^-$. (e) and (f) show the
characteristics for $v^+$ and $v^-$.  The horizontal line at time
$t=0.1$ represents the remapping of the simulation. In both (e) and
(f) the characteristics are converging to the shocks, while never
intersecting.} \label{example3b}
\end{center}
\end{figure}

\subsection{Comparisons between the homentropic Euler and the CAHE equations}
From these examples and previous sections we can begin to draw
comparisons between the homentropic Euler equations and the CAHE
equations.

Both sets of equations are capable of forming shocks and expansion
waves from the Riemann problem.  For homentropic Euler equations it
is well established that for the Riemann problem there will often be
both a shock and an expansion wave formed.  With the CAHE equations,
introduced in this paper, there will always be two shocks formed.
However, if the initial conditions are slightly smoother, some of
those discontinuities will be found to be unstable and result in an
expansion wave.  Thus both behaviors can be said to be found in both
equations.

In terms of speed of the shocks, we have found there to be
differences.  In the homentropic Euler equations, the speed of the
shocks are determined by the Rankine-Hugoniot conditions, in order
to preserve mass and momentum.  With the CAHE equations the speed of
a discontinuity is determined by the speed of the averaged
characteristics at the location of the discontinuity.  In section
\ref{shockspeedsection}, it can be seen that with the averaging
chosen for this paper that the shock speed differs from those of the
homentropic Euler equations.  From this it is also clear that while
the homentropic Euler equations may be formally regained from the
CAHE equations by letting $\alpha \to 0$, the solution will not
converge to weak solutions of the homentropic Euler equations.  It
may be possible, however, to choose an averaging scheme such that
the shock speeds between the two equations are the same or similar.

Additionally the conditions under which a single traveling shock
differs between the two equations as is demonstrated in Examples 2
and 3.  For a single shock to form with the CAHE equations either
$v^+$ or $v^-$ must be constant.  Even if it were possible to find
an averaging scheme such that the shock speeds between the two
equations are the same, these conditions would not change.  For the
homentropic Euler equations, the condition is tied in with the
Rankine-Hugoniot conditions.

The final difference noted between these equations is significant.
As seen in Examples 2 and 3, for the homentropic Euler equations,
there can be areas devoid of characteristics.  These can be filled
with expansion waves or by new characteristics originating from the
shock.  By averaging the characteristics in the CAHE equations,
there will be no areas devoid of characteristics.  Thus there will
be no generation of new characteristics.  This can cause significant
differences in behavior as seen in Example 2, where $v^-$ remains a
constant for the CAHE equations, but not for the homentropic Euler
equations.  This again is a property that is unaffected by the
averaging scheme chosen.

\section{Conclusions}

The CAHE equations were derived by taking the homentropic Euler
equations and spatially averaging the characteristics. This led to a
new set of equations that has many interesting properties. Existence
and uniqueness proofs, Theorem \ref{existencetheorem} and
\ref{existencetheorem2}, were proven by establishing that the
characteristics of the equations never intersect.  In the process
this established that any initial conditions in $C^1(\mathbb{R})$
will have a solution in $C^1(\mathbb{R})$. Furthermore any
discontinuities in the initial conditions will remain and be
convected in the solution for all time.

The speeds of shocks in the CAHE equations were found to be
determined by the speed of the characteristics at the location of
those shocks.  Furthermore, different averaging schemes were shown
to provide different shock speeds.

The Riemann problem was then examined, where solutions were
generally found to consist of two traveling discontinuities.
However, often one of the discontinuities can prove to be unstable
and by smoothing the initial conditions will develop into an
expansion wave instead.

Finally using some numerical examples and results from the previous
sections the CAHE equations were compared and contrasted with the
homentropic Euler equation from which they originated.  Both can
generate shocks and expansion waves from the Riemann problem. With
the averaging scheme employed in the paper, the speed of the shocks
differ.  For the examples chosen, the behavior of the equations
proved to be significantly different.  Finally, the homentropic
Euler equations showed the generation of new characteristics, while
with the CAHE equations, new characteristics will never be
generated.  From this is seems clear that as $\alpha \to 0$ the
solutions to the CAHE equations will not converge to weak solutions
of the homentropic Euler equations.

The CAHE equations have proven to have convenient existence and
uniqueness properties.  When continuous initial conditions are
chosen, the solution will remain continuous.  However, the equations
seem to display behavior with too significant of departure from the
homentropic Euler equations to be of use in gas dynamic
applications. It is possible that this could be rectified by taken a
one-sided average of the characteristics, but this concept was not
pursued here.

\section{Acknowledgments}
The research in this paper was partially supported by the AFOSR
contract FA9550-05-1-0334.

\bibliography{../RefA1}
\bibliographystyle{unsrt}

\end{document}